\documentclass[sigconf, nonacm]{acmart}
\usepackage[binary-units]{siunitx}
\usepackage{makecell}
\usepackage{mathtools}
\usepackage{multirow}
\usepackage{simplebnf}
\usepackage{algorithm}
\usepackage{todonotes}
\usepackage{float}
\usepackage[inline]{enumitem}
\usepackage{amsmath}
\usepackage{lipsum} 
\usepackage{subcaption}
\usepackage{algorithmicx,algpseudocode}
\usepackage{tikz}
\usetikzlibrary{automata, positioning, shapes.geometric, arrows.meta}
\newcommand\vldbdoi{XX.XX/XXX.XX}

\newcommand\vldbavailabilityurl{URL_TO_YOUR_ARTIFACTS}
 
\usepackage{xspace,amstext,url}

\newif\ifdraft\drafttrue
\ifdraft
\newcommand\anthony[1]{{\color{blue}
\small [#1 - \textbf{Anthony}]}}
\newcommand\domagoj[1]{{\color{magenta}
\small [#1 - \textbf{Domagoj}]}}
\newcommand\heyang[1]{{\color{brown}
\small [#1 - \textbf{Heyang}]}}

\newcommand\heyangchanged[1]{{{#1}}}
\else
\newcommand\anthony[1]{}
\newcommand\domagoj[1]{}
\newcommand\heyang[1]{}

\newcommand\heyangchanged[1]{#1}
\newcommand\todo[1]{}
\fi

\newcommand\pgraph{$G = (V, E, \rho, \lambda, \pi)$}
\newcommand\paut{$\mathcal{A} = (\Sigma, \chi, Attr, Q, Q_0, F, \Delta)$}
\newcommand\query{$\mathrm{PRPQ}(v, pregex)$}

\newcommand{\Vertices}{\ensuremath{\mathsf{Vertices}}\xspace}

\newcommand{\Edges}{\ensuremath{\mathsf{Edges}}\xspace}

\newcommand{\Lab}{\ensuremath{\mathsf{Labels}}\xspace}

\newcommand{\Values}{\ensuremath{\mathsf{Values}}\xspace}
\newcommand{\Prop}{\ensuremath{\mathsf{Properties}}\xspace}

\newcommand{\semanticsfull}[2]{[\! [#1]\!]_{#2}}
\newcommand{\semanticsele}[2]{[\! [#1]\!]^{#2}}

\newcommand{\Q}{\mathbb{Q}}

\begin{document}
\title{Answering Constraint Path Queries over Graphs}

\author{Heyang Li}
\affiliation{%
  \institution{University of Kaiserslautern-Landau}
  \city{Kaiserslautern}
  \state{Germany}
}
\email{heyang.li@cs.rptu.de}

\author{Anthony Widjaja Lin}
\orcid{0000-0003-4715-5096}
\affiliation{%
    \institution{University of Kaiserslautern-Landau and MPI-SWS}
  \city{Kaiserslautern}
  \country{Germany}
}
\email{awlin@mpi-sws.org}

\author{Domagoj Vrgo\v{c}}
\affiliation{%
  \institution{PUC Chile and IMFD Chile}
  \city{Santiago}
  \country{Chile}
}
\email{vrdomagoj@uc.cl}




\begin{abstract}
    Constraints are powerful declarative constructs that allow users to
    conveniently restrict variable values that potentially range over an
    infinite domain. In this paper, we propose a constraint path query language
    over property graphs,
    which extends Regular Path Queries (RPQs) with SMT constraints on data
    attributes in the form of equality constraints and Linear 
    Real Arithmetic (LRA) constraints. We provide efficient algorithms 
    for evaluating such path queries over property graphs, which exploits 
    optimization of macro-states (among others, using theory-specific 
    techniques). 
    In particular, we demonstrate how such an algorithm may effectively utilize 
    highly optimized SMT solvers for resolving such constraints over paths. 
    We implement our algorithm in MillenniumDB, an open-source graph engine
    supporting property graph queries and GQL. Our extensive empirical
    evaluation in a real-world setting demonstrates the viability of our 
    approach.
\end{abstract}

\maketitle


\ifdefempty{\vldbavailabilityurl}{}{
\vspace{.3cm}
\begingroup\small\noindent\raggedright\textbf{Artifact Availability:}\\
The source code, data, and/or other artifacts have been made available at \url{https://github.com/Wanshuiquan/MillenniumDB/tree/artifact}.
\endgroup
}

\section{Introduction}
\label{sec:intro}
Graph databases have been an increasingly popular technology in the database
ecosystem over the past decades, with multiple
open-source~\cite{JenaTDB,kuzu,mdb} and proprietary
systems~\cite{Webber12,Oracle,memgraph,nebula,TigerGraph} being developed, and a
steady stream of research literature on the
subject~\cite{survey,Baeza13,DeutschFGHLLLMM22}. Graphs offer an intuitive
modelling of the application domain with nodes representing entities and edges representing 
connections between these entities. Edges are usually labelled to identify the
type of the connection, giving rise to the edge-labelled graph database model~\cite{CruzMW87,Baeza13}. Extending this model with the ability to add attributes with their associated values to both nodes and edges (and also to label the nodes) is supported in the property-graph data model~\cite{survey}.

While the early work on graph databases focused on edge-labelled graphs, partly due to their widespread use in the Semantic Web community and the availability of standards such as RDF~\cite{RDF} and SPARQL~\cite{HarrisS13}, the main focus of commercial vendors these days are property graphs. Historically the main player in this space has been Neo4j with their Cypher query language~\cite{cypherpaper}, with many vendors implementing their variants Cypher. Many different flavours of the query language also meant low interoperability, so significant amount of efforts was put into standardizing property-graph query languages by the ISO/IEC resulting in SQL/PGQ~\cite{sql-pgq-standard}, and GQL~\cite{gql-standard} standards for querying property graphs.

At their core, both standards share the same pattern matching capabilities, starting from graph patterns~\cite{survey}, which allow finding a small graph-like pattern inside of a larger property graph. Another core feature of all graph query languages are \emph{path queries}, which allow traversing the graph with paths whose length is not know in advance. Traditionally~\cite{HarrisS13,survey,CruzMW87}, the main class of such queries were \emph{regular path queries (RPQs)}, which were specified via a regular expression and would return pairs of nodes connected by a path whose edge labels spell a word in the language of the expression.

    \begin{figure}[h]
        \centering
    \resizebox{0.6\linewidth}{!}{%

\begin{tikzpicture}[
    node distance=4.2cm,
    person/.style={circle, draw=black, thick, minimum size=2cm, align=center, font=\Huge},
    arrow/.style={-{Stealth[length=3mm]}, draw=black, thick},
    label/.style={font=\Huge, midway, fill=white, inner sep=1pt, align=center}
]

\node[person] (alice) {Alice \\ age: 25};
\node[person, right=of alice] (bob) {Bob \\ age: 30};
\node[person, below=of alice] (charlie) {Charlie \\ age: 28};
\node[person, right=of charlie] (diana) {Diana \\ age: 32};

\draw[arrow, bend left=15] (alice) to node[label, above] {\texttt{follow}, since: 2020} (bob);
\draw[arrow, bend left=15] (bob) to node[label, above]{\texttt{follow}} node[label, below] { since: 2019} (alice);

\draw[arrow, bend left=15] (charlie) to node[label, above] {\texttt{follow}, since: 2021} (diana);
\draw[arrow, bend left=15] (diana) to node[label, below] {\texttt{follow}, since: 2018} (charlie);

\draw[arrow] (alice) to node[label, left, xshift=-3pt] {\texttt{follow},\\ since: 2022} (charlie);

\draw[arrow, bend right=15] (bob) to node[label, above, sloped] {\texttt{favorite}} node[label,below, sloped]{since: 2020}(charlie);
\draw[arrow, bend right=15] (diana) to node[label, right, xshift=2pt] {\texttt{favorite}, \\since: 2019} (bob);

\node[above=0.1cm of alice, font=\Huge] {\texttt{Person}};
\node[above=0.1cm of bob, font=\Huge] {\texttt{Person}};
\node[below=0.1cm of charlie, font=\Huge] {\texttt{Person}};
\node[below=0.1cm of diana, font=\Huge] {\texttt{Person}};

\end{tikzpicture}
    }
        \caption{Social network property graph.}
        \Description{A diagram of a property graph representing a social network, with person nodes and follow/favorite edges.}
        \label{fig:property_graph}
    \end{figure}
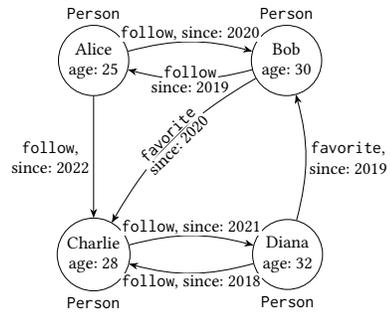

In GQL and SQL/PGQ path queries were significantly enhanced compared to previous efforts, allowing regular path queries to return different sort of paths between two nodes (shortest, simple, etc.), and there is an ongoing effort to include reasoning on complex path properties such as length, maximum values, or cost  into the GQL standard. To illustrate the importance of such features, consider the social network graph in Figure~\ref{fig:property_graph} representing information about people who know each other. Suppose now that we wish to find paths labelled by \texttt{follows} such that the link is rather new (e.g. established after 2021), but also that the maximum difference between the age of two people along this paths is no more than seven years. In our example the path \texttt{Alice}$\rightarrow$\texttt{Charlie}$\rightarrow$\texttt{Diana} is one such path connecting \texttt{Alice} to \texttt{Diana}. While such queries are relevant in practice, not many existing engines can specify them or execute them efficiently, be it because they lack full support for RPQs~\cite{FariasMRV-iswc24}, or because they cannot express complex data properties such as the maximum age gap in the example above~\cite{cypherpaper,LibkinMV-jacm16}.

\paragraph*{Constraints} The aforementioned query is an example of 
\emph{constraint queries} \cite{constraint-databases} specifically applied to 
graph databases. That is, one allows variables over a possibly infinite domain 
(e.g. the set of real numbers), which could be constrained by using formulas 
over certain logical theories (e.g. Linear Real Arithmetic (LRA)). Constraints
were studied systematically in database theory in the 1990s resulting in
several prototypes including DEDALE \cite{dedale}, MLPQ \cite{mlpq}, and DISCO
\cite{disco}, although such systems were limited to databases not exceeding
hundreds of tuples. To the best of our knowledge, none of these systems are 
still maintained and available in the public domain. Recently, constraints have
been revisited in the context of graph databases \cite{FJL22,FLP25,LSY25}. In
particular, a large class of constraint path queries can be answered efficiently
(i.e. in nondeterministic logarithmic space and polynomial time, for any fixed
query). Unfortunately,
all of these algorithms rely on heavy machinery from embedded finite model
theory called \emph{Restricted Quantifier Collapse (RQC)}. In fact, the
proposed algorithms rewrite a given constraint query into a simpler one that is
in the worst case \emph{doubly exponentially larger} than the original query!

\paragraph*{Contributions} In this paper, we demonstrate that it is possible to
enrich RPQs with complex data constraints, without sacrificing efficiency and
scalability of query evaluation. Specifically, our contributions can be 
summarized as follows:
%
%
\begin{itemize}
\item We introduce \emph{parametric regular expressions}, which provide a clean
    syntax to specify path patterns which: (a) conform to a regular expression;
        (b) allow defining complex data constraints on attribute values along 
        such paths. In particular, constraints of the form of (dis)equality
        over strings and existential (in)equality formulas over linear terms over rational 
        variables are permitted.
\item We show that parametric regular expressions can be converted into
    \emph{parametric automata} \cite{seq-theory,FJL22,FL22}, which may be 
        construed
        as a subclass of Regular Data Path Queries \cite{FJL22}, where
        variables are ``read-only'' but they are allowed to take values that are 
        not in the database. 
        Among others, this allows an extension of the product-graph 
        construction used to evaluate RPQs~\cite{FariasMRV-iswc24}, yielding
        a constraint reachability problem over property graphs.
    \item We provide new lightweight query evaluation algorithms, which do not
        use heavy machinery from constraint databases and embedded finite model
        theory (in particular, RQC \cite{FJL22,FLP25}).
        The new algorithm is essentially a graph reachability algorithm over
        ``macro-states'' (i.e. a data structure consisting of a node in the
        graph, a state in the automaton, and a set of accumulated constraints).
        In particular, a simplex algorithm (which is supported by most 
        SMT-solvers) can be used to efficiently determine feasibility of
        a macro-state. The algorithm runs in $O(2^c \cdot |A| \cdot (|V| + |E|))$ time, 
        where the query has size $A$ with $c$ (in)equalities constraints,
        and the property graph has $|V|$ vertices and $|E|$ edges. The
        algorithm is exponential only in the size of the query, which is 
        unavoidable owing to our NP-hardness of the problem. 
        Since the query is typically much smaller than the database, 
        we may use \emph{data complexity} \cite{Vardi82} to measure the 
        complexity of the algorithm (i.e., $c$ and $|A|$ as a constant size),
        in which case our algorithm runs in linear time.
    \item 
        We implement our new query evaluation algorithms inside of 
        MillenniumDB~\cite{mdb}, an open-source graph engine supporting 
        property graph queries and GQL; and
    \item We provide an extensive experimental evaluation showing the 
        feasibility of our approach in a real-world setting, \emph{up to tens of
        millions of tuples}. This is far beyond the database size that previous
        constraint database systems \cite{dedale,disco,mlpq} could handle, i.e.,
        up to hundreds of tuples. On average, our approach can
        evaluate most queries within \SI{100}{ms} for medium dense graphs
        and complicated queries over extremely dense graphs 
        within \SI{1}{s}. This performance is despite the NP-hard combined 
        complexity of the problem.

\end{itemize}

\paragraph*{Related work} While path queries that constrain how data values
change along a path conforming to a regular expression (or extensions thereof
for constraining data) have been studied in the
theoretical literature~\cite{LibkinMV-jacm16,FJL22,BFL15}, to the best of our
knowledge, not much work was done on actually evaluating such queries in
practice. This might not be surprising, given that such data conditions mimic 
aggregation over paths, which is a provably hard problem~\cite{FrancisGGLMMMPR23}. However, such queries are highly relevant in practice, and are being actively added to the current version of the GQL standard~\cite{DeutschFGHLLLMM22}. For this same reason, certain path constraints are supported in existing systems through the \texttt{UNWIND} operator~\cite{Webber12}, which allows to first collect all the values along a path and then subsequently process then as a list, the performance in such cases seems to be somewhat lacking~\cite{GheerbrantLR25}, which is to be expected given that the number of paths matching the underlying regular pattern can easily become exponential~\cite{FariasMRV-iswc24} and the \texttt{UNWIND} approach requires collecting all of them for post-processing. Even when such approaches are efficient, they do not provide a systematic way of expressing path constraints as parameterized regex we introduce do, since the latter can support data constraints under a regular pattern.

\paragraph*{Organization} The remainder of the paper is organized as follows. Preliminary definitions are given in 
Section~\ref{sec:prelim}. We define parametric RPQs and their automaton model in Section~\ref{sec:language}. Algorithms 
for their evaluation are given in Section~\ref{sec:naive-algo} and Section~\ref{sec:algos}. Experimental evaluation is in Section~\ref{sec:exp}.
 We conclude in Section~\ref{sec:concl}.

\section{Preliminaries}
\label{sec:prelim}
Here we define basic notions used throughout the paper.

\paragraph{Graphs and paths}  Following the usual conventions in the research literature
~\cite{DBLP:conf/sigmod/DeutschFGHLLMPSVZ22}, we define \emph{property graphs} as directed graphs where both edges and nodes carry labels and a series of attributes (i.e. properties) with their associated values. Formally, we assume disjoint countably infinite sets \Vertices of vertex identifiers, \Edges of edge identifiers, and \Lab of edge labels. Similarly, we assume a countably infinite set \Prop of node and edge property names and \Values of property values. We can then define property graphs as follows:
\begin{definition}
    \label{def:property_graph}
    A \textit{property graph} is a tuple $G = (V, E, \rho, \lambda, \pi)$ where:
\begin{enumerate*}
    \item $V\subseteq \Vertices$ is a finite set of vertex identifiers;
    \item $E\subseteq \\\Edges$ is a finite set of edge identifiers;
    \item $\rho: E \rightarrow  \left(V \times V \right)$ is a total function mapping edges to 
    ordered pairs of vertices. \heyangchanged{For convenience, instead of writing $\rho\left(e\right) =  \left(v_1, v_2\right)$, 
we shall often write $v_1\xrightarrow{e} v_2$;}
    \item $\lambda : \left(V \cup E \right) \rightarrow \Lab$ is a total function assigning a label to a vertex or an edge; and
    \item $\pi: \left(V \cup E\right) \times \Prop \rightarrow  \Values$ is a partial function mapping an element (edge or vertex) and a property name to a property value.
\end{enumerate*}
\end{definition}
\begin{example}
    \label{ex:property_graph}
    Consider the property graph $G$ depicted in Figure~\ref{fig:property_graph} modeling a social network, where vertices represent people, and edges represent 
    relationships between them. The labels of vertices is \texttt{Person}, while edges are labeled as \texttt{follow} or \texttt{favorite}. Each person vertex has 
    a \texttt{name} property and an \texttt{age} property. Each edge has a \texttt{since} property indicating the year when the relation started.
\end{example}
\heyangchanged{
\begin{definition}[Path]
    \label{def:path}
    A \emph{path} from  $v$ to $w$ in a property graph~$G=(V, E, \rho, \lambda, \pi)$, where $v, w  \in V$, is an alternating sequence $v = v_0 e_1 v_1 e_2, \ldots, e_n v_n = w$ of vertices and edges
    where $n \ge 0$, such that a path can be a unit vertex $v$ where $n = 0$, or for all $i \in [1,n]$,    $e_i \in E$ and $\rho\left(e_i\right) = \left(v_{i-1}, v_{i}\right)$ for \emph{forward edge} or $\rho\left(e_i\right) = \left(v_{i}, v_{i-1}\right)$ for \emph{backward edge}.
\end{definition}

Note: a path can be empty, denoted by $\epsilon$. 
}

\paragraph{Regular path queries} Regular path queries (RPQs for short)~\cite{Baeza13} in a graph database $G$ is an expression of the form $(v, \texttt{regex},?x)$, where $v$ is a node in $G$ and \texttt{regex} is a regular over the alphabet of edge labels, and $?x$ the output variable. The output of an RPQ over $G$, denoted $\semanticsfull{(v, \texttt{regex},?x)}{G}$, is the set of all nodes $v'$ such that $v'$ can be reached from $v$ by a path $P$ in $G$ and the edge labels along this path form a word accepted by the regular expression \texttt{regex}.

\paragraph{The product graph construction} A common way to evaluate RPQs is based on the product graph construction~\cite{CruzMW87,FariasMRV-iswc24}. Given a graph database $G = (V,E,\rho,\lambda,\pi)$ and an RPQ $q = (v,\texttt{regex},\texttt{?x})$, the \emph{product graph} is constructed by first converting the regular expression \texttt{regex} into an equivalent non-deterministic finite automaton $(Q,\Sigma,\delta,q_0,F)$. Here $Q$ is a finite set of states, $\Sigma$ a finite alphabet of edge labels, $\delta \subseteq Q\times \Sigma\times Q$ the transition relation, and the initial  state is $q_0$, while $F$ is the set of final states. The product graph $G_\times$ is defined as the graph database $G_\times = (V_\times, E_\times, \rho_\times, \lambda_\times,\pi_\times)$, where
\begin{enumerate*}
    \item $V_\times = V\times Q$;
    \item $E_\times = \{(e,(q_1,a,q_2)) \in E\times \delta \mid  \lambda(e) = a\}$;
    \item $\rho_\times(e,d) = ((x,q_1),(y,q_2))$ if: $d = (q_1,a,q_2)$, $\lambda(e) = a$ and  $\rho(e)=(x,y)$;
    
    \item $\lambda_\times((e,d)) = \lambda(e)$;
    \item $\pi_\times(v,q) = \pi(v)$;
    \item $\pi_\times(e,d) \\=  \pi(e)$.
\end{enumerate*}
In the final two items we abuse the notation slightly to signal that the set of attributes for nodes or edges is inherited from the original database $G$. Each node of the form $(u,q)$ in $G_\times$ corresponds to the node $u$ in $G$ and, furthermore, each path $P$ of the form $(v,q_0)(e_1,d_1)(v_1,q_1)\dots (e_n,d_n)(v_n,q_n)$ in $G_\times$ corresponds to a path $p = v e_1 v_1 \dots e_n v_n$ in $G$ that (a) has the same length as $P$ and (b) brings the automaton from state $q_0$ to $q_n$.
As such, when $q_n \in F$, then this path in $G$ matches $\texttt{regex}$. In other words, all nodes $v'$ that can be reached from $v$ by a path that matches $\texttt{regex}$ can be found by using standard graph search algorithms (e.g., BFS/DFS) on $G_\times$ starting in the node $(v,q_0)$.

\section{Query Language}
\label{sec:language}

This section introduces the syntax and semantics of \emph{parametric regular expressions} and \emph{parametric regular path queries}. Parametric regular 
expressions extend standard regular expressions (used in RPQs and GQL) in two aspects: (i) parametric regular expressions can not only express patterns over edges,
but also express patterns over nodes; and (ii) parametric regular expressions have \emph{constraints} on data domains of edges and nodes, and these constraints can 
query beyond the active domain of a property graph with \emph{global parameters}. In addition to parametric regular expressions, we also introduce \emph{parametric automata},
which serve as an execution model for parametric regular path queries. Finally, we analyze the complexity of evaluating parametric regular path queries over property graphs.

\subsection{Parametric Regular Expressions}

\subsubsection*{Syntax} The syntax of parametric regular expressions is based on the syntax of regular expressions in~\cite{mdb}, and the formal definition is given below.

\begin{definition}[Syntax of Parametric Regular Expressions]
\label{def:syntax-para-regex}
We assume a set $\mathcal{L}$ of labels, a finite set $\mathcal{P}_n$ of \emph{numerical} properties, a finite set $\mathcal{P}_c$ of \emph{string} properties 
and a set of \emph{global parameters} $Var$ such that $Var \subset \Q$. The syntax of parametric regular expressions is defined as follows:
\begin{bnfgrammar}
$E$::= $(t, \phi)$: $t \in \mathcal{L}$
|  $\string^E$: inverse
| $E_1 / E_2$: concatenation of $E_1$ and $E_2$ 
| $E_1 \mid E_2$: alternation of $E_1$ and $E_2$
| $E^*$
| $E^+$
| $E^{?}$
;;


$\phi$::=  $x = c_{str}$: $x \in \mathcal{P}_c$ and $c_{str}$ is a string constant
| $t_{ar} \sim t_{ar}$:$\sim \in \{>, < , \le, \ge, \neq, =\}$
| $\phi \land \phi$
;;
$t_{ar}$::=  $c$: $c \in \mathbb{Q}$
| $c\cdot p$: $c \in \mathbb{Q}$ 
| $t_{ar} + t_{ar}$ 
;;
$p$::= $attr$: $attr \in \mathcal{P}_n$
| $?x$: $?x \in Var$
\end{bnfgrammar}
\end{definition}

For the convenience of presentation, we call $attr \in \mathcal{P}_n$ as \emph{numerical attributes}, and $attr \in \mathcal{P}_c$ as
\emph{string attributes}. 


\subsubsection*{Semantics} 




We start with the \emph{sequence of elements}, which is what can be captured by a parametric regular expression. We assume a countably 
infinite set \Lab~of \emph{forward} labels, \Prop~of properties names and \Values~of properties values. For each forward label $l \in \Lab$, we have a related \emph{inverse label} $\string^l$, and 
the set of inverse labels is denoted as $\string^\Lab$ An element is a pair $(l, f)$, such that $l \in \Lab \cup \string^\Lab$ and $f: \Prop \mapsto \Values$. 
We denote a set of elements as $\Sigma$. 

A \emph{sequence of elements} is a finite list $p$ such that $p \in \Sigma^*$. Given two sequences $p_1, p_2 \in \Sigma^*$, where $p_1 = (l_1, f_1), \dots, (l_n, f_n)$ and 
$p_2 = (l_1', f_1'), \dots, (l_n', f_n')$, the \emph{concatenation} of $p_1$ and $p_2$ is a new sequence $p_1 \cdot p_2 = (l_1, f_1),  \dots, (l_n, f_n), (l_1', f_1'), \dots, 
(l_n', f_n')$.

\paragraph{Interpretation of Data Constraints} 
The data constraints in parametric regular expressions are interpreted in terms of $assignments$ to global parameters and the theory of Quantifier-Free Linear Real Arithmetic and 
the theory of equality (with string constants)~\cite{Bradley-Book} which are supported by major SMT solvers like Z3~\cite{10.1007/978-3-540-78800-3_24}.
In the following, we denote the theory of Quantifier-Free Linear Real Arithmetic and equality as $\mathfrak{T}$.

A global parameters assignment is a total function that assigns each global parameter to a rational value, and we write $\emptyset$ for empty assignments. 
 Given two assignments $\mu$ and $\mu'$, we say that $\mu$ and $\mu'$ are \emph{unifiable} if $\mu(x) = \mu'(x)$ for all $x \in dom(\mu) \cap dom(\mu')$, and
 we define the \emph{unification} $(\mu \cup \mu')$ of two unifiable assignment $\mu, \mu'$ $(\mu \cup \mu')(x) = \mu(x)$ if $x \in dom(\mu)$, and otherwise 
 $(\mu \cup \mu')(x) = \mu'(x)$.

\begin{definition}[Semantics of Data Constraint] 
Given an element $(l, f)$, an assignment to global parameters $\mu$ and a data constraints $\phi(\overline{x}, \overline{p})$ where $\overline{x}$ are global parameters and $\overline{p}$ are properties 
variables. $\mathfrak{T} \models \phi(\mu(\overline{x}), f(\overline{p}))$ is defined inductively over $\phi$ as follows:
\begin{itemize}
    \item $\mathfrak{T} \models \theta_1 \land \theta_2$ iff $\mathfrak{T} \models \theta_1$ and $\mathfrak{T} \models \theta_2$.
    \item $\mathfrak{T} \models f(p) = c_{str}$ iff $p$ is a string property variable and $f(p) = c_{str}$.
    \item $\mathfrak{T} \models \theta(\mu(\overline{x}), f(\overline{p}))$ iff $\overline{p}$ are numeric properties variables and $\theta(\overline{x}, \overline{p})$ is an (in)equality over linear terms over rational variables.
\end{itemize}
\end{definition}

\begin{definition}[Language of Parametric Regular Expression]
    \label{def:lanuguage-pregex}
 The \emph{language} of a parametric regular expression $r$ with an assignment $\mu$ is a set of sequences over $\Sigma$, 
, denoted as $\semanticsele{r}{\mu} \subseteq \Sigma^*$. $\semanticsele{r}{\mu}$ is defined inductively as follows:

\begin{align*}
    \semanticsele{\epsilon}{\emptyset} &= \{\epsilon\} \\
  \semanticsele{(l, \phi(\overline{x}, \overline{p}))}{\mu} &= 
  \left\{
    (l,f) 
    \,\middle|\,
    \begin{array}{@{}l@{}}
        \mathfrak{T} \models \phi(\mu(\overline{x}), f(\overline{p}))
    \end{array}
  \right\}  \\ 
   \semanticsele{\string^r}{\mu} &=  \left\{
    (\string^l, f)
   \,\middle|\,   
    \begin{array}{@{}l@{}}
        (l, f) \in \semanticsele{r}{\mu}
    \end{array}
  \right\} \\
  \semanticsele{r_1 / r_2}{\mu_1\cup \mu_2} &= \left\{(  a_1 \cdot a_2) \;\middle|\;  a_1 \in \semanticsele{r_1}{\mu_1}, a_2 \in \semanticsele{r_2}{\mu_2}
  \right\}  \\
  \semanticsele{r_1 \mid r_2}{\mu_1\cup \mu_2} &= \left\{ a \;\middle|\; 
  \begin{array}{@{}l@{}}
 a \in \semanticsele{r_1}{\mu_1} \cup \semanticsele{r_2}{\mu_2}
  \end{array}
  \right\} 
\end{align*}
Moreover, assuming that $r^1 = r$  and $r^{n+1} = r^n /r$ for every $n \ge 1$, we have:
\begin{align*}
    \semanticsele{r^*}{\emptyset \cup \mu_1 \cup \cdots  \mu_i} =  \semanticsele{\epsilon}{\emptyset} \cup \bigcup_{n \ge 1} \semanticsele{r^n}{\mu_i}
\end{align*}
where each $\mu_i, \mu_j, 0 \le i,j \le n$ are pairwise unifiable.

The definition of $r^+$ can be derived from $r/r^*$, and $r^?$ can be derived from $r \mid \epsilon$. 
\end{definition}

\begin{definition}[Membership of Parametric Regular Expressions]
    \label{def:membership}
    We say that $r$ accepts an element sequence $e$ with an assignment $\mu$, as long as $e \in \semanticsele{r}{\mu}$.
    The acceptance is denoted by $p \in \mathcal{L}(r)$
\end{definition}

\paragraph{Semantics over Property Graphs} The semantics of parametric regular expressions is coherent with the semantics of path patterns over property graphs~\cite{mdb}.
We take each path $p$ over a property graph~\pgraph~as a sequence of elements $(l_1, \pi'_1), \dots, \\(l_n, \pi'_n)$,  where for each $p_i \in E \cup V, l_i = \lambda(p_i)$, 
and $\pi': \Prop \mapsto \Values$ is a partial function such that $\pi'(a) = \pi(p_i, a)$. Note the inverse-labeled element in a path $p$ are interpreted as nodes or 
inverse edges over $G$, i.e. if $p_i \in p$ such that $p_i = (l_i, \pi'_i)$ and $l_i = \string^l$, then either $p_i \in V$ and $l = \lambda(p_i)$, or $p_i = \rho(p_{i+1}, p_{i-1})$.

The \emph{set of answers} of a parametric regular expression $r$ over a property graph~\pgraph, denoted as $\semanticsfull{r}{G}$ is a set of paths,
such that each path $p \in \mathcal{L}(r)$.

\subsection{Parametric Regular Path Queries}
  

\begin{definition}[Parametric Path Regular Queries]
\label{def:pquery}
A parametric regular path query ($\mathrm{PRPQ}$ for short) over a property graph~\pgraph\ is an expression of the form $q = (v, pregex)$, with $v\in V$, 
and $pregex$ a \emph{parametric regular expression}. The query $q$ over a property graph~\pgraph\ returns \texttt{true} if there exists a node $v'$, 
a path $p$ from $v$ to  $v'$, such that $p \in \semanticsfull{pregex}{G}$. Otherwise the query returns \texttt{false}.
\end{definition}

\begin{example}
    Now we give an example of parametric regular path query. Let us consider the property graph in example~\ref{ex:property_graph}, and assume we want to check whether there is a path using \texttt{follow} 
    edges, such that \emph{the distance between the maximal and the minimal age is smaller than 7, and all \texttt{follow} connections should start in 2019}. This can be formalized as the following 
    parametric regular expression $r_f$:
    \label{ex:prpq}
    \begin{align*}
        &((\texttt{Person}, ?p \le age \land ?q \ge age \land ?q - ?p \le 7)/\\&(\texttt{follow}, since > 2019 ))^*/\\
        &(\texttt{Person}, ?p \le age \land ?q \ge age \land ?q - ?p \le 7)
    \end{align*}
    
    If we denote the node representing `Alice' as $v_0$, and we can formalize the above problem as $\mathrm{PRPQ}(v_0, r_f)$. 
\end{example}

\subsection{Parametric Automaton} 

We propose \emph{parametric automaton} as our execution model inspired by
sequence theory~\cite{seq-theory} which decide the constraints over a sequence of objects 
with data domains. Parametric automata take an input path as a sequence of objects with attributes, capture object patterns and evaluate data conditions and find a model 
for global parameters. We will show that parametric automata have the same computation power as parametric regular expressions.

\begin{definition}[Parametric Automaton]
    \label{def:automaton}
A parametric automaton is a tuple $(\mathsf{L}, \chi, Attr, Q, Q_0, F, \Delta)$, where: 
\begin{enumerate*}
    \item $\mathsf{L} \subseteq \Lab$ is a set of symbols
    \item $\chi$ is a set of global variables.
    \item $ Attr \subseteq \Prop$ is a set of properties.
    \item $Q$ is a finite set of states 
    \item $Q_0 \subseteq Q$ is a set of start states. 
    \item $F \subseteq Q$ is a set of final states. 
    \item We define the transitions $\Delta$ as follows:  
    $\Delta \subseteq Q \times (\Sigma \times T(\chi, Attr) \times \mathbb{B}) \times Q  $ where $T(\chi, Attr)$ is the set of string and linear arithmetic formulas over
    $\chi \cup Attr$ following the format of Definition~\ref{def:syntax-para-regex}, and $\mathbb{B}$ is a boolean value that indicates the transition is inverse if it is true and vice versa.

    If $ (q, (\sigma, \phi, inv) ,q') \in \Delta$, where $\sigma \in \mathsf{L}, \phi \in T(\chi, Attr)$,  then we write $q \xrightarrow{\left( \sigma, \phi, inv \right)} q'$. 
\end{enumerate*}
\end{definition}




\begin{definition}[Acceptance Conditions of Parametric Automaton]
    \label{def:accept}

A parametric automaton $Aut = (\mathsf{L}, \chi, Attr, Q, Q_0, F, \Delta)$.
    $Aut$ \emph{accepts} an element sequence $(l_1, f_1), \dots, (l_n, f_n) \in \Sigma^*$ with respect to 
an assignment $\mu$, if there is a sequence of transitions 
\[
q_0 \xrightarrow{\left( \sigma_1, \phi_1, inv_1 \right)} q_1 \xrightarrow{\left( \sigma_2, \phi_2, inv_2 \right)} q_2 \ldots \xrightarrow{\left( \sigma_n, \phi_n, inv_n \right)} q_n
\] 
such that: \begin{enumerate*}
\item  $q_0 \in Q_0, q_n \in F$ 
\item   $\forall i \in \{1, \dots, n \}, \sigma_i = l_i$;
\item  $\forall i \in \{1, \\ \dots, n \}, \mathfrak{T} \models \phi_i(\mu(\overline{x}), f_i(\overline{p}))$.
\item if $inv_i$ is true, $l_i$ should be an \emph{inverse} label, and otherwise $l_i$ should be a \emph{forward} label.
\end{enumerate*}
\end{definition}


\begin{example}
    The parametric regular expression in example~\ref{ex:prpq} can be formalized as the 
    following parametric automaton $Aut$ depicted in figure~\ref{fig:automaton-example}.
  \begin{figure}[htbp]
    \centering
 \resizebox{0.6\linewidth}{!}{%

\begin{tikzpicture}[
    ->, >=Stealth,
    node distance=2.5cm,
    every state/.style={thick, fill=gray!10, minimum size=0.8cm},
    initial text={}
]

\node[state, initial] (q0) {$q_0$};
\node[state, right of=q0] (q1) {$q_1$};
\node[state, right of=q1] (q2) {$q_2$};

\node[state, below of=q0] (q3) {$q_3$};

\node[state, accepting, below of=q1] (qf) {$q_f$};

\draw (q0) edge node[above] {\texttt{Person}} node[below]{$\varphi$} (q1);
\draw (q1) edge node[above] {\texttt{follow}} node[below]{$since > 2019$} (q2);
\draw (q2) edge[bend left=30] node[left] {\texttt{Person}} node[below]{$\varphi$} (qf);
\draw(qf) edge[bend left=30] node[above] {\texttt{follow}} node[below]{$since > 2018$} (q3);
\draw(q3) edge[bend left=30] node[above] {\texttt{Person}} node[below]{$\varphi$} (qf);
\end{tikzpicture}
    }    \caption{Parametric automaton example, where {$\varphi := ?p \le age \land ?q \ge age \land ?q - ?p \le 7 $};}
\label{fig:automaton-example}
  \end{figure}
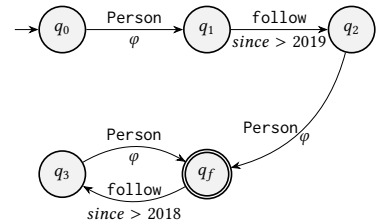
\end{example}

Each parametric regular expression can be converted to a parametric automata. The following conversion is a variant of 
a standard regular expression to NFA conversion. For completeness, we present the conversion. 

\begin{theorem}
    For each parametric regular expressions $r$,  there is a parametric automaton \paut with a single initial state such that
    for each $p \in \semanticsele{r}{\mu}$, $r$ is accepted by $\mathcal{A}$ with assignment $\mu$.
\end{theorem}

\begin{proof}
    We prove by structural induction on parametric regular expressions.
    \paragraph{Base case:}
    For atomic expression $E =  (t, \phi)$, the parametric automaton of  $E$ is $Aut = (\mathsf{L}, \chi, Attr, Q, \{q_0\}, F, \Delta)$, where: 
\begin{enumerate*} 
    \item $\mathsf{L} = \{t\}$ 
    \item $\chi$ is all the global parameters in $\phi$.
    \item $Attr$ is all the properties in $\phi$.
    \item $Q = \{q_0, q_f\}$ 
    \item $\{q_0\}$ is the only start state. 
    \item $F = \{q_f\}$ is a set of final states. 
    \item $\Delta = \{(q_0, (t, \phi, false) ,q_f)\}$, 
\end{enumerate*}
    
We can verify each $p \in \semanticsele{(t, \phi)}{\mu}$ accepted by $Aut$ according to Definition~\ref{def:lanuguage-pregex}.

\paragraph{Induction Step}:
For a concatenation expression $E_1 / E_2$, if $E_1$ is translated to  $Aut_1 = (\mathsf{L}_1, \chi_1, Attr_1, Q_1, \{q_{01}\}, F_1, \Delta_1)$, 
and $E_2$ is translated to $Aut_2 = (\mathsf{L}_2, \chi_2, Attr_2, Q_2, \{q_{02}\}, F_2, \Delta_2) $, we construct a parametric automaton 
$Aut' = (\Sigma, \chi, Attr, Q, \{q_01\}, F', \Delta')$, with  $F = F_2$.  we redirect all transition towards states in $F1$ to $q_{02}$. Formally
    \begin{align*}
        \Delta' &= \left(\Delta_1 \setminus \{(q, (t, \phi, inv), q_f), \mid q_f \in F_1 \}\right) \cup \Delta_2 \\
               &\cup \{ (q, (t, \phi, inv), q_{02}) \mid \forall (q, (t, \phi, inv), q_1) \in \Delta_1 , q_1 \in F_1\}
    \end{align*}
For each $p' \in \semanticsele{E_1 \cdot E_2}{\mu}$, we have $p' = p_1 \cdot p_2$ and $\mu' = \mu_1 \cup \mu_2$ where
$p_1 \in \semanticsele{E_1}{\mu_1}$ and $ p_2 \in \semanticsele{E_2}{\mu_2}$. By induction hypothesis, $p_1$ is accepted by $Aut_1$ with 
$\mu_1$, and $p_2$ is accepted by $Aut_2$ with $\mu_2$. Then $e$ has an accept run with the new transition in $\Delta'$ and the unified 
assignment $\mu'$.

For an alternation expression $E_1 \mid E_2$, $E_1$ is translated to an automaton 
$Aut_1 = (\mathsf{L}_1, \chi_1, Attr_1, Q_1, \{q_{01} \}, F_1, \Delta_1) $, and $E_2$ is translated to an automaton $Aut_2 = (\mathsf{L}_2, \chi_2, Attr_2, Q_2, \{q_{02} \}, F_2, 
\Delta_2) $.  We construct $Aut' = (\Sigma, \chi, Q, \Sigma, {q_0'}, F', \Delta')$ by introducing a new initial state $q_0$ and discarding $q_{01}$ and $q_{02}$. The final 
states set is $F' = F_1 \cup F_2$, and let all transitions from $q_{01}$ and $q_{02}$ start from $q_0'$, formally 
     \begin{align*}
           \Delta' &= \left(\Delta_1 \setminus \{(q_{01}, (t, \phi, inv), q) \mid q \in Q_1\}\right) \\
           &\cup \left(\Delta_2 \setminus \{(q_{02},(t, \phi, inv), q) \mid q  \in Q_2\}\right) \\
           &\cup \{(q_0', (t, \phi, inv), q) \mid (q_{01}, (t, \phi, inv), q) \in \Delta_1\} \\
           &\cup \{ (q_0', (t, \phi, inv), q) \mid (q_{02}, (t, \phi, inv), q) \in \Delta_2\}
        \end{align*}

For each $p' \in \semanticsele{E_1 \mid E_2}{\mu}$, we have $p' = p_1$ or $p' = p_2$ and $\mu' = \mu_1 \cup \mu_2$ where
$ p_1 \in \semanticsele{E_1}{\mu}$ and $ p_2 \in \semanticsele{E_2}{\mu}$. By induction hypothesis, $p_1$ is accepted by $Aut_1$ with 
$\mu_1$, and $p_2$ is accepted by $Aut_2$ with $\mu_2$. Then $p'$ has an accept run with the new transition from $q'_{0}$ with $\mu'$.
   

For a Kleene star expression $E^*$, if $E$ is translated to $Aut = (\mathsf{L}, \chi , Attr, Q, \{q_{0}\}, F, \Delta) $,  the parametric automaton of 
$E^*$ is $Aut' = (\mathsf{L}, \chi, Attr, Q', \{q_0\}, F', \Delta')$ with $Q' = (Q \setminus F)$ and  $F' = \{q_0\}$, and  we redirect all transitions 
in $\Delta_f$ to $q_0$, i.e.
    \begin{align*}
        \Delta' &= \left(\Delta \setminus \Delta_f\right) \\
               &\cup \{(q, (t, \phi, inv), q_0) \mid (q, (t, \phi, inv), q_f) \in \Delta_f\} 
    \end{align*}

For each $p' \in \semanticsele{E^*}{\mu}$, we have $p' \in \semanticsele{E^i}{\mu_i}$ where $i > 0$ or $p' = \epsilon$ with empty assignment.
For the $\epsilon$ case, since the initial state is also the final state, then $\epsilon$ is accepted without any conditions. 
If we have $p' \in \semanticsele{E^i}{\mu_i}$,  by induction hypothesis, $p'$ is accepted by $Aut$ with $\mu_i$, and then $p$ has an accept run towards
$q_{0}$ with assignment $\mu$.




For an inverse expression $\string^E$, and $E$ can be translated to $Aut = (\mathsf{L} , \chi,Attr, Q, \{q_{0}\}, F, \Delta) $, it is suffices to modify 
all transitions of $\Delta$, such that $\Delta' =  \{(q, (t, \phi, !inv), q') \mid (q', (t, \phi, inv), q) \in \Delta\}$. For each $p \in \semanticsele{\string^E}{\mu}$,
according to Definition~\ref{def:lanuguage-pregex}, $p$ satisfies the condition of inverse transitions.

Since $?$ and $+$ are derivable from $\mid$ and $*$, it is enough to show the above basic operators. 
\end{proof}


\subsection{The Hardness of Parametric Regular Path Queries}
Although the data complexity (i.e. the complexity of evaluating a fixed \query~query over a property graph $G$) is $\mathsf{NL}$-complete \cite{FJL22,Sipser-book},
the combined complexity (where both the query and $G$ are part of the input) is $\mathsf{NP}$-hard \emph{in general}. This can be shown by a reduction from the 
Boolean satisfiability problem (3-SAT).

\begin{theorem}
The combined complexity of parametric regular path queries over a property graph is $\mathsf{NP}$-hard.
\label{theorem:np_hardness}
\end{theorem}

\begin{proof}
We reduce from the Boolean satisfiability problem (3-SAT). Given a Boolean formula $\phi$ in CNF with variables $x_1, x_2, \ldots, x_n$ and clauses $C_1, C_2, \ldots, C_m$, 
we construct a property graph $G$ and a \query~query such that $\phi$ is satisfiable if and only if there exists a path in $G$. 

We construct $G$ with a single vertex $v$ 
labeled by $v$ with a single property $a$ such that $\pi(v, a) = 1$, and a single self loop edge $e$ labeled by $e$ with no properties and $\rho(e) = (v,v)$.

We construct a parametric regular expression based on the formula $\phi$ \emph{inductively}, with a set of rational global parameters $\{?x_1, \dots, \\?x_n\}$. The base 
case starts from the first clause $C_1$, where we construct a parametric regular expression $r_1$ as follows:

\[
r_1 = (v, ?x_{i} \sim_i 0) \mid (v, ?x_{j} \sim_j 0) \mid (v, ?x_{k} \sim_k 0)
\]

if $?x_i$ is positive in $C_1$, then $\sim_i$ is $\neq$ otherwise $\sim_i$ is $= $. The cases for $?x_j, ?x_k$ are similar. 

If we encode $C_1, \dots C_i$ by a parametric regular expression $r_i$, we construct $r_{i+1}$ for clause $C_{i+1}$ as follows:
\begin{align*}
  r_{i+1} &=  r_i/(((e, true)/(v, ?x_{i} \sim_i 0) \\
    &\mid (e, true)/(v, ?x_{j} \sim_j 0) \\
    &\mid (e, true)/(v, ?x_{k} \sim_k 0)))
\end{align*}

where $\sim_i, \sim_j, \sim_k$ are defined as in the base case. 

According to the Definition~\ref{def:pquery}, we construct the parametric regular path query as $\text{PRPQ}(v, r_m)$, and if there exists a path in $G$ that matches $r_m$, 
then the corresponding assignment of global parameters $\{?x_1, \dots, ?x_n\}$ in $r_m$ also satisfies all clauses in $\phi$. Conversely, if $\phi$ is satisfiable, the model 
of all clauses provides an assignment to the global parameters $\{?x_1, \dots, ?x_n\}$ in $r_m$, and then there exists a path in $G$ that matches $r_m$. 

Therefore, we have reduced 3-SAT to the problem of evaluating a \query~query over a property graph, proving that the combined complexity is $\mathsf{NP}$-hard.
\end{proof}

\section{Naive Evaluation Algorithm}
\label{sec:naive-algo}
Algorithm~\ref{alg:naive_algorithm} is a straight-forward naive algorithm for $\text{PRPQ}(v, r)$ over a property graph \pgraph based on the 
SIMPLE path semantics~\cite{Farias2023evaluating} i.e.,the paths do not repeat any node. Algorithm~\ref{alg:naive_algorithm} 
returns a pair $(\texttt{true}, (p, \mu))$ if $p \in \mathcal{L}(r)$ is detected, and otherwise
returns $(\texttt{false}, \emptyset)$. Algorithm~\ref{alg:naive_algorithm} compiles $r$ into a parametric automaton $\mathcal{A}$, explores $G$ by 
\emph{synchronized breadth-first search} and accumulating all visited data constraints. As long as the algorithm reaches a final state in $\mathcal{A}$, 
the algorithm checks accumulated constraints by an SMT solver, and constructs a path by backtracking if the checking results are satisfied.


\subsubsection*{Synchronized Transitions}
Synchronized BFS is a variant of search on the product graph, because parametric regular expressions 
include patterns on both edges and nodes. Synchronized BFS is based on synchronized transitions.
During traversal, whenever a new node $v'$ or edge $e'$ is encountered in $G$, the algorithm matches the element's labels 
against $\mathcal{A}$'s transition conditions, advances the automaton state accordingly, and accumulates the corresponding formulas.

When the algorithm visits a node $v$ in property graph $G$ and the `current' location in $\mathcal{A}$ is $q$, 
the algorithm \emph{makes a node synchronized transition} by \textsc{TransNode} function. {$\textsc{TransNode}(v, q, G, \mathcal{A})$} 
scans all transitions $q \xrightarrow{(l, \varphi, inv)} q'$ of $\mathcal{A}$ that originate from $q$, checks whether the 
transition label $l$ matches the label of $v$, and returns all satisfied successor states $q'$.

When the algorithm explores the neighbor edges from a node $v$ in $G$, the algorithm \emph{makes an edge synchronized transition} by \textsc{TransEdge} function.
{$\text{TransEdge}(e, q, G, \mathcal{A})$} scans all transitions $q \xrightarrow{(l, \varphi, inv)} q'$ that originate from $q$ and
matches with edges by case analysis on the $inv$ flag. If the $inv$ flag is $\top$, the function collects all edges $e'$ \emph{originate from} $v$ with $\lambda(e) = l$, and returns
with successor states and the destination node $v'$ pairwise. If the $inv$ flag is $\bot$, the function collect all edges $e'$ \emph{enter into} $v$ with $\lambda(e) = l$, and returns with 
successor states and the source node $v$ pairwise.

Formally, the synchronized transition functions are defined as:

\footnotesize{
\begin{align*}
  \label{help-fun-naive}
  \textsc{TransNode}(v, q, G, \mathcal{A}) &=  \{q' \mid q \xrightarrow{(l, \varphi, inv)} q' , l = \lambda(v) \} \\ 
  \textsc{TransEdge}(e, q, G, \mathcal{A}) &=  \begin{cases*}
     \{(v', e', q') \mid q \xrightarrow{(l, \varphi, \top)} q', \string^l = \string^(\lambda(e')), v \xrightarrow{e'} v' \} \\
      \{(v', e', q') \mid q \xrightarrow{(l, \varphi, \bot)} q', l = \lambda(e'), v' \xrightarrow{e'} v \} 
  \end{cases*}
\end{align*}
}
\normalsize

After a synchronized transition $q \xrightarrow{(l, \varphi, inv)} q'$ with an object $o \in V \cup E$, the algorithm substitutes the attributes in $\varphi$ with 
actual property values from $o$, and accumulate the resulting instantiated formula in a set of visited formulas. We formalize the accumulation 
by the following function:
\[
\textsc{Update}(q, q', o, \mathcal{A}, G, F) =  F \cup \{\varphi[attr/\pi(o, attr)] \mid q \xrightarrow{(l, \varphi, inv)} q'\}
\]

The naive algorithm manipulates \emph{search states} to record the above information. Each search state represents a snapshot containing: the current position in both the 
automaton and property graph ($q$ and $v$ respectively), the traversed edge ($e$), a reference to the preceding state ($prev$), and the accumulated formulas ($\mathcal{F}$).

\begin{definition}
  A search state is a tuple $(v, q, e, prev, \mathcal{F})$ where:
\begin{enumerate*}
\item $v \in V$ is the current graph node
\item $q \in Q$ is the current automaton state
\item $e \in E$ is the current graph edge. 
\item $prev$ is a pointer to the previous search state (enabling path reconstruction via backtracking~\cite{Farias2023evaluating})
\item $\mathcal{F}$ contains formulas with attribute variables instantiated using actual property values from visited nodes/edges
\end{enumerate*}
\label{def:search_state}
\end{definition}

\subsubsection*{Query Algorithm} The \textsc{NaiveQuery} procedure in Algorithm~\ref{alg:naive_algorithm} is the main procedure to evaluate a parametric regular query, which 
requires the following data structures:
\begin{enumerate*}
  \item Open, which is a queue of search states, with usual queue operations (enqueue, dequeue).
  \item Visited, which is a dictionary of search states using tuples $(v,e,q)$ as keys, that have already been explored, maintained to avoid infinite loops, and the visited node 
  can be collected by function $\textsc{VisitedNode}(Visited) = \{v \mid (v,e,q) \in Visited.keys()\}$

\end{enumerate*}


The procedure begins exploration from the initial state $(v, q_0)$. First, it checks whether a trivial path only containing $v \in V$ constitutes 
a valid answer (Lines 5-8), and such case does not require a model of global parameters. The procedure then performs an initial synchronization with $(v, q_0)$(Lines 8-12): 
executing one transition in $\mathcal{A}$ starting from $v$, accumulating the corresponding data constraints, and enqueuing the resulting search states into the Open queue.

The main loop (Lines 13-35) processes states until either the `Open' queue is exhausted, for each iteration:
(1) Dequeues a search state $(v, q, e, prev, \mathcal{F})$ from Open (Line 14), and skip repeated nodes  (Lines 16-18).
(2) Performs two synchronized transitions in $\mathcal{A}$, first with the outgoing edge $e'$ from $v$ by \textsc{TransEdge} (Line 15), and second with the adjacent 
node $v'$ connected via $e'$ by \textsc{TransNode} (Line 20). The algorithm accumulates instantiated data constraints after each transition by \textsc{Update} (Lines 19 and 21).
(3) Tests if the resulting state $(v', q')$ has been visited before (Line 23), if not, the algorithm enqueues the resulting search states into the Open queue and 
adds to the Visited dictionary.
(4) If the automaton state $q'$ is a final state, the algorithm checks whether the accumulated formulas $\mathcal{F}$ are satisfiable. If these conditions are satisfied, the algorithm constructs the path by the \textsc{GetPath} procedure via backtracking on 
the $prev$ domain of search states, and constructs a model of the global parameters by querying an SMT solver, and then the algorithm returns \texttt{true} with the path and
the model.

If the `Open' queue is exhausted, and no answer is detected, then the algorithm returns \texttt{false}.



\begin{algorithm}
\caption{Query evaluation of~\query~over a property graph~\pgraph.} 
  \label{alg:naive_algorithm}                                                                                                                                                                                                                                                                                                                                                                                                                                                                                                                                                                                                                                                                                                                                                                                                                                                                                                                                                                                                                                                                                                                                                                                                                                                                                                                                                                                                                                                                                                                                                                                                                                                                                                                                                                                                                                                                                                                                                                                                                                                                                                                                                                                                                                                                                                                                                                                                                                                                                                                                                                                                                                                                                                                                                                                                                                                                                                                                                                                                                                                                                                                                                                                                                                                                                                                                                                                                                                                                                                                                                                                                                                                                                                                                                                                                                                                                                                                                                                                                                                                                                                                                                                                                                                                                                                                                                                                                                                                                                                                                                                                                                                                                                                                                                                                                                                                                                                                                                                                                                                                                                                                                                                                                                                                                                                                                                                                                                                                                                                                                                
\begin{algorithmic}[1]
      \Procedure{NaiveQuery}{$regex, G, v$}
            \State $\mathcal{A} \gets Automaton(regex)$ \Comment{$q_0$ initial state, $F$ final states} 
            \State Open.init()
            \State Visited.init()
        
            \If{$v \in V$ \textbf{and} $q_0 \in F$}
                \State \textbf{return}(\texttt{true}, ($v$, $\emptyset$)) \Comment{No need for a model}
            \EndIf
            \ForAll {$q'$ $\in$ \Call{TransNode}{$v, q_0, G, \mathcal{A}$}}
                \State $\mathcal{F}$ $\gets$ \Call{Update}{$q_0, q', v, \mathcal{A}, G, \emptyset$}
                \State start\_state $\gets$ $(v, q', \bot, \bot, \mathcal{F})$
                \State Open.enqueue(start\_state)
            \EndFor
            \While {Open $\neq \emptyset$}
            \State $\text{state}=(v, q, e, prev, \mathcal{F})$ $\gets$ Open.dequeue() 
            \ForAll {$(v', e', q') \in$ \Call{TransEdge}{$e, q, G, \mathcal{A}$}}
            \If {$v' \in \textsc{VisitedNode}(\text{Visited})$}
             \State continue \Comment{Skip not simple path}
            \EndIf
              \State $\mathcal{F}'$ $\gets$ \Call{Update}{$q, q', e' ,\mathcal{A}, G, \mathcal{F}$}
              \ForAll  {$ q'' \in$ \Call{TransNode}{$v', q', G, \mathcal{A}$}}
                \State $\mathcal{F}''$ $\gets$ \Call{Update}{$q', q'', v', \mathcal{A}, G, \mathcal{F}'$}
                \State nextState $\gets$ $(v', q'', e', curr, \mathcal{F}'')$
                \If {$(v' ,q'' , e', *, *)$ $\notin$ Visited}
                    \State Visited.add(nextState)
                    \State Open.enqueue(nextState)
                    \If {$q'' \in F$ \textbf{and} \Call{CheckSat}{$\mathcal{F}''$}} 
                        \State path $\gets$ \Call{GetPath}{nextState}
                        \State $\mu$ $\gets$ \Call{GetModel}{$\mathcal{F}''$}
                        \State \textbf{return} (\texttt{true}, (path, $\mu$))
                    \EndIf
                \EndIf
              \EndFor
            \EndFor
        \EndWhile
        \State \textbf{return} (\texttt{false}, $\emptyset$) \Comment{No path detected}
      \EndProcedure

      \Procedure{GetPath}{state = $(v, q, e, prev, \mathcal{F})$}
      \If {$prev = \bot$}  
          \State \textbf{return} v
      \Else
          \State \textbf{return} \Call{GetPath}{$prev$}.extend(e, v)
      \EndIf
      \EndProcedure

\end{algorithmic}
\end{algorithm}




\section{Optimized Query Evaluation}
\label{sec:algos}


Although Theorem~\ref{theorem:np_hardness} states that the combined complexity of parametric regular path queries is $\mathsf{NP}$-hard, Algorithm~\ref{alg:query_algorithm} 
presents a \emph{feasible algorithm} based on \emph{macro states}, to replace accumulated formulas with bounds of global parameters 

terms stored in \emph{macro states}.
\begin{definition}[Macro State]
  \label{def:macro_state}
   A macro state is a tuple $(q, e, v, prev, \\up, low, neq)$ where: 
    \begin{enumerate*}
        \item $q$~,~$e$~,~$v$~,~$prev$ are the same as in a search state according to Definition~\ref{def:search_state}.  
        \item $up$ stores the upper bounds of terms. Formally, for $\sum a_i p_i \le c$, 
        where $a_i, c \in \mathbb{Q}$ and $p_i \in \chi$, if and only if $up\left[\sum a_i p_i\right] = c$
        \item $low$ stores the lower bounds of terms. Formally, for $\sum a_i p_i \ge c$, where $a_i, c \in \mathbb{Q}$ and 
        $p_i \in \chi$, if and only if $low\left[\sum a_i p_i\right] = c$
        \item $neq$ handles the $\neq$ operand.  Formally, for $\sum a_i p_i \neq c$, where $a_i, c \in \mathbb{Q}$ and 
        $p_i \in \chi$, if and only if $neq\left[\sum a_i p_i\right] = c$
\end{enumerate*} 
\end{definition}

As macro states only permit upper and lower bounds, we introduces a new constant $\epsilon$ with $\epsilon > 0$, and rewrite the (in)equalities of each data constraint 
by applying the following rewriting rules:
\begin{align}
  \begin{aligned}
     t_l < t_r &\mapsto t_l + \epsilon \le t_r \\
     t_l > t_r &\mapsto t_l - \epsilon \ge t_r \\
     t_l = t_r &\mapsto t_l \le t_r \land t_l \ge t_r 
  \end{aligned}
  \label{rewrite}
\end{align}

The rewriting introduces $O(n)$ new (in)equalities over linear terms over rational variables, where $n$ is the number of (in)equalities in original formula.

Algorithm~\ref{alg:query_algorithm} takes the synchronized transitions framework as the naive algorithm with the same \textsc{TransNode} and \textsc{TransEdge} procedures, and
the main change is to leverage a new procedure to update a macro state and check the consistency of bounds by querying an \emph{oracle}, which is formalized in the \textsc{NewUpdate}
procedure in Algorithm~\ref{alg:update_macro_state}.

\begin{algorithm}
\caption{Update the bounds of terms $(up, low)$ with an object $o$ of a property graph $G$ from $q$ to $q'$ in a parametric automaton $\mathcal{A}$}
  \label{alg:update_macro_state}

\begin{algorithmic}[1]
      \Procedure{NewUpdate}{$q, q', o, A, G, (up, low)$} 
            \State $\varphi' \gets \varphi[attr/\pi(o, attr)]~\text{where}~q \xrightarrow{(l, \varphi, inv)} q'$
            \ForAll {$atom \in \varphi'$} \Comment{$\varphi'$ is a conjunction of atoms}
            \If {$atom$ is string constraint \textbf{and} not satisfiable}
                \State \textbf{return} $\bot, \emptyset$ \Comment{Inconsistency detected}
            \EndIf
            \State Normalize $atom$ into $\sum a_i p_i \sim c$ where $\sim \in \{\le, \ge, \neq\}$ 
            \If {$\sim~\text{is}~\le$}
                \If {$\sum a_i p_i \notin up$ \textbf{or} $c \le up[\sum a_i p_i]$}
                    \State $up[\sum a_i p_i] \gets c$
                \EndIf
            \ElsIf {$\sim~\text{is}~\ge$}
                \If {$\sum a_i p_i \notin low$ \textbf{or} $c \ge low[\sum a_i p_i]$}
                    \State $low[\sum a_i p_i] \gets c$
                \EndIf 
             \ElsIf {$\sim~\text{is}~\neq$} 
                  \State $neq[\sum a_i p_i] \gets c$
                \EndIf
              \EndFor
              \If{\Call{CheckSat}{up, low, neq}} \Comment{Query the oracle to check consistency}
                  \State \textbf{return} $\top, (up, low, neq)$ \Comment{Update successful}
              \Else
                  \State \textbf{return} $\bot, \emptyset$ \Comment{Inconsistency detected}
              \EndIf
            \EndProcedure
\end{algorithmic}
\end{algorithm}




\begin{algorithm}
\caption{Optimized Query Evaluation of~{\small\query} over a property graph~{\small\pgraph}}
\label{alg:query_algorithm}
\begin{algorithmic}[1]
      \Procedure{OptimizedQuery}{$regex, G, v$}
            \State $\mathcal{A} \gets Automaton(regex)$ \Comment{$q_0$ initial state, $F$ final states}
            \State Rewrite all data constraints in $\mathcal{A}$ according to (\ref{rewrite})
            \State Open.init()
            \State Visited.init()
        
            \If{$v \in V$ \textbf{and} $q_0 \in F$}
                \State \textbf{return}(\texttt{true}, ($v, \emptyset$))
            \EndIf
            \ForAll {$q'$ $\in$ \Call{TransNode}{$v, q_0, G, \mathcal{A}$}}
              \State $init, b_0 \gets$ \Call{NewUpdate}{$q_0, q', v, \mathcal{A}, G, (\emptyset, \emptyset, \emptyset)$}
              \If {$init$}
                \State $(up_0, low_0, neq_0) \gets b_0$
                \State start\_state $\gets$ $(v, q', \bot, \bot, up_0, low_0, neq_0)$
                \State Open.enqueue(start\_state)
                \EndIf
            \EndFor
            \While {Open $\neq \emptyset$}
            \State $\text{state}=(v, q, e, prev, up, low, neq)$ $\gets$ Open.dequeue() 
            \ForAll  {$(v' , e', q') \in$ \Call{TransEdge}{$e, q, G, \mathcal{A}$}}
            \If{$v' \in \textsc{VisitedNode}(\text{Visited})$}
            \State continue \Comment{skip no simple path}
            \EndIf
              \State $b \gets (up, low, neq)$
              \State $flag, b' \gets$ \Call{NewUpdate}{$q, q', e', \mathcal{A}, G, b$}
              \If { \textbf{not} $flag$} 
                    \State continue \Comment{inconsistency detected}
                \EndIf
              \ForAll {$ q'' \in$ \Call{TransNode}{$v', q', G, \mathcal{A}$}}
                \State $flag', b'' \gets$ \Call{NewUpdate}{$q', q'', v', \mathcal{A}, G, b'$}
                \If { \textbf{not} $flag'$} 
                    \State continue \Comment{inconsistency detected}
                \EndIf
                \State (up', low') $\gets$ b''
                \State nextState $\gets$ $(v', q'', e', curr, up', low', neq')$
                \If {$(v' ,q'' , e', *, *, *, *)$ $\notin$ Visited}
                    \State Visited.add(nextState)
                    \State Open.enqueue(nextState)
                    \If {$q'' \in F$} 
                        \State path $\gets$ \Call{GetPath}{nextState}
                        \State $\mu$ $\gets$ \Call{GetModel}{$up', low', neq$}
                        \State \textbf{return}(\texttt{true}, (path, $\mu$))
                    \EndIf
                \EndIf
              \EndFor
            \EndFor
        \EndWhile
        \State \textbf{return} (\texttt{false}, $\emptyset$) \Comment{No path detected}
        \EndProcedure
\end{algorithmic}
\end{algorithm}

\paragraph{Complexity Analysis} However, Algorithm~\ref{alg:query_algorithm} queries the linear programming solver $O(2^c \cdot |A| \cdot (|V| + |E|))$ times, where 
$\mathcal{A}$ is a parametric automaton with $c$ (in)equality constraints. The $2^c$ factor arises because each of the $c$ constraints 
can be independently either present or absent in a macro state, and the algorithm must explore all such combinations in the worst case.


\section{Experimental evaluation}
\label{sec:exp}

In this section, we present an experimental evaluation of parametric regular path queries using the algorithms presented in Section~\ref{sec:naive-algo} 
and Section~\ref{sec:algos}. Given that, to the best of our knowledge, no other system supports parametric regular path queries, we focus on showing the efficiency of our algorithms over real world graphs. In particular, we focus on showing that the optimized version of our approach, presented in Algorithm~\ref{alg:query_algorithm} is a \emph{feasible} solution over large graphs, and that its macro-state based optimization achieves significant improvement of the baseline solution given in Algorithm~\ref{alg:naive_algorithm}. 
For this, we focus on the following research questions:
\begin{enumerate}[label=\textbf{RQ\arabic*}]
\item How does Algorithm~\ref{alg:query_algorithm} scale with graph size and density?
\item How does the performance vary with different queries depending both on regular expressions and data constraints?
\item How does the oracle query affect the overall performance?
\item What is the performance improvement of the optimized algorithm (Algorithm~\ref{alg:query_algorithm}) over the naive algorithm (Algorithm~\ref{alg:naive_algorithm})?
\end{enumerate}


\subsection{Experiment Setup}

\paragraph{Implementation}
We implement parametric regular path queries and their evaluation algorithms in \textsc{MillenniumDB}~\cite{mdb}, an open-source persistent graph database system written using the \verb|C++| programming language. We use the \textsc{Z3} SMT solver~\cite{10.1007/978-3-540-78800-3_24} as oracle for 
queries. The syntax of parametric regular expression is embedded into the MQL query syntax used in \textsc{MillenniumDB}, 
which resembles Cypher and GQL. Both the naive query algorithm and the optimized query algorithm are incorporated into the standard execution pipeline of \textsc{MillenniumDB}.

\paragraph{Datasets} Table~\ref{tab:dataset} lists the characteristics of the labeled-graph dataset used in the experiments. 
The ICIJ-Leaks and ICIJ-Paradises datasets~\cite{icij_offshoreLeaks, icij_Paradisepapers} contain information about offshore entities and their relationships, which 
have natural attributes and labels on both nodes and edges. These datasets are relatively sparse.
The LDBC01 and LDBC10 datasets are originally used for graph database benchmarking~\cite{Lissandrini:2018:GDB}, which are synthetic datasets and simulate
a social network with various types of nodes and edges. These datasets are relatively densely connected, and the LDBC10 dataset is significantly larger than
the other datasets used in the experiments.
The Pokec dataset~\cite{takac2012Pokeclargegraph} is a social network. Nodes represent users and contain many attributes including age, gender, 
and location, while edges represent relations between users. Labels on nodes and edges are generated synthetically. Pokec dataset has a large size and higher density.
The Telecom dataset~\cite{snapnets} contains the relationships between users and behaviors in a telecom network, which has natural attributes and labels on both nodes and edges. 
This dataset is also the most dense graph among the datasets we used.

\begin{table}[ht]
  \caption{Dataset Characteristics.($1\,\mathrm{M} = 10^6, 1\,\mathrm{K} = 10^3$)}
  \label{tab:dataset}
  \small
  \begin{tabular}{ccccccc}
    \toprule
     Name&Dataset & $|V|$ & $|E|$ & $|Le|$ & $|Lv|$ & $\frac{|E|}{|V|}$\\
    \midrule
    IL&ICIJ-Leaks &1.9M & 3.2M & 14 & 5 & $1.68 $  \\
    IP&ICIJ-Paradises & 163K & 364K & 6 & 5 & $2.23$  \\
    L0&LDBC01 & 180K & 768K &  8&  15&$4.17$ \\
    L1&LDBC10 & 30M & 178M & 9 & 15 & $5.93$ \\
    PO&Pokec & 1.6M & 30.6M & 3 & 1 & $18.7$ \\
    TE&Telecom & 170K & 50M & 3 & 4  & $294$ \\

  \bottomrule
\end{tabular}
\end{table}

\paragraph{Query Generation} Query templates in our experiments consist of two components: \emph{regular path templates} which are regular expressions that represent patterns of paths, and 
\emph{data constraint templates} that describe relations between objects along a path. A \emph{query template} is instantiated by incorporating a data constraint into a regular path template. A concrete query is then constructed by: (i) incorporating a start node into the query (ii) replacing label placeholders in path templates with concrete edge labels; and (iii) replacing attribute name placeholders in data constraint templates with concrete attribute names. Next we describe each component in detail.


Table~\ref{tab:queries} lists the 12 regular path templates used in the experiments. These are the top-12 most frequent property-path query patterns occurring in practice based on the study of publicly available query logs for SPARQL endpoints~\cite{DBLP:journals/pvldb/BonifatiMT17}. We use regular expressions in Table~\ref{tab:queries} as regular path  templates in the experiments, and fix $k=3$. Notice that edge labels here (e.g. $a_1, a_2$, etc.) are abstract placeholders which get instantiated with concrete labels.

\begin{table}[hb]
    \caption{Regular templates used in the experiments. }
  \label{tab:queries}
  \begin{tabular}{c c | c c}
    \toprule
    Name & Type & Name & Type \\
    \midrule
    $Q_1$ & $(a_1 | \cdots |a_k)^*$  & $Q_7$ & $a_1?/\cdots/a_k?$\\
    $Q_2$  & $a^*$ & $Q_8$ & $a/(b_1| \cdots |b_k)$\\
   $Q_3$ & $a_1/\cdots/a_k$ & $Q_9$ &  $a_1/a_2?/\cdots/a_k?$ \\
   $Q_4$ & $a^*/b$ & $Q_{10}$ & $(a/b^*)|c$ \\ 
   $Q_5$ & $a_1|\cdots|a_k$ & $Q_{11}$ &$a^*/b?$\\
    $Q_6$ & $a^+$ & $Q_{12}$ & $a/b/c^*$\\
  \bottomrule
\end{tabular}
\end{table}

We classify regular path templates into three categories based on their occurrence in query logs (see Table 5 in  in~\cite{DBLP:journals/pvldb/BonifatiMT17}). The results of this classification are listed in Table~\ref{tab:query_categories}.

\begin{table}[htbp]
  \centering
  \caption{Categories of regular templates used in the experiments.}
  \label{tab:query_categories}
  \small
  \begin{tabular}{c|c|c|c}
    \toprule
    Category & \makecell{Regular \\Templates} & \makecell{Relative \\Occurrence} & \makecell{Total \\Percentage}\\
    \midrule
    Frequently-used & $Q_1$, $Q_2$, $Q_3$, $Q_4$ & >10\% & 87.58\%\\
    \hline
    Occasionally-used & $Q_5$, $Q_6$, $Q_7$ & >1\% and <10\% &12.34\%\\
    \hline
    Rarely-used & \makecell{$Q_8$, $Q_9$, $Q_{10}$,\\ $Q_{11}$, $Q_{12}$} & >0.01\% and <1\% & 0.08\%\\
    \bottomrule
  \end{tabular}
  
\end{table}

Table~\ref{tab:data_constraint} (column labelled ``Description'') lists the descriptions of 5 data constraint templates we use, including two \emph{simple arithmetic constraints} $D_1$ and $D_2$, and three \emph{complex arithmetic constraints} $D_3$, $D_4$ and $D_5$. Combining the five data constraint templates with the twelve regular path templates, gives us a total of 60 query templates. A concrete query is then constructed in two phases.  First, we replace the label placeholders in path templates with concrete edge labels, and replace attribute name placeholders in data constraint templates with concrete attribute names, giving rise to a parametric regular expression \texttt{pregex}. We then select a starting node $v$ for our parametric regular path query 
$(v,\texttt{pregex})$ (see Definition~\ref{def:pquery}). An example of a partially instantiated query template (changing only the edge labels) is given in Table~\ref{tab:data_constraint} (rightmost column).
               The source vertices are chosen randomly from the graph, and replacing label placeholders with labels selected from most frequent edge labels occurring in the graph. These 100 concrete query instances per each template are used to evaluate the naive algorithm, 
               the optimized algorithm, and we also use a version which removes the data constraints to compare the impact of incorporating data constraints into RPQs. 
               The latter serves as a baseline which is supported by \textsc{MillenniumDB}, so it allows us to do a fair comparison.


\begin{table*}[htbp]
  \centering
  \caption{Data constraints used in the experiment. }
  \label{tab:data_constraint}
  \begin{tabular}{c |c| c| c}
    \toprule 
    Name & Description & Example with $a^*$ & Category\\
    \midrule
    D1 & \makecell{The distance between the average of attribute $attr$\\ and the values of $attr$ should within a threshold $c$} & \makecell{\( \begin{array}{rcl} &\{a_v~, ?p - attr \le c \land attr - ?p \le c \}/ \\ &(\{a_e, \top\} / \{a_v~, ?p - attr \le c \land attr - ?p \le c\})^* \end{array} \)} & Simple\\
    \midrule
    D2 & \makecell{Upper and lower bound of an attribute $attr$} & \makecell{\( \begin{array}{rcl} &\{a_v~, ?p \le attr \land ?q \ge attr\}/ \\ &(\{a_e, \top\} / \{a_v~, ?p \le attr \land ?q \ge attr\})^* \end{array} \)} & Simple \\
    \midrule
    D3 & \makecell{The distance of upper and lower bound of attribute $attr$\\ should be within threshold $c$.} & \makecell{\(
      \begin{array}{rcl}
        &\{a_v~, ?p \le attr \land ?q \ge attr \land ?q - ?p \le c\}/ \\
        &(\{a_e, \top\} / \{a_v~, ?p \le attr \land ?q \ge attr \land ?q - ?p \le c\})^*
      \end{array}
    \)} & Complex\\
    \midrule
    D4 & \makecell{Let the value of attribute $attr_1$ of the start point be $a_{10}$,~\\and the value of attribute $attr_2$ of the start point be $a_{20}$\\
    for each successor along a path,\\ $0.5\cdot  a_{10} + c_1 \le attr_1$ and $|a_{20} - attr_2| \le c_2$  } & 
    \makecell{\( \begin{array}{rcl} &\{a_v~, ?p = attr_1 \land ?q = attr_2 \}/ \\ &(\{a_e, \top\} / \{a_v~, ?p \cdot 0.5 + c_1 \le attr_1 \land  ?q - attr_2 \le c_2  \\&\land attr_2 - ?q \le c_2\})^* \end{array} \)} & Complex\\
    \midrule
    D5 & \makecell{Let the two-dimensional manhattan distance between \\the start point $(x_1, y_1)$ and  each node along a path \\$(x_2, y_2)$ be within a threshold $c$.} & 
    \makecell{\( \begin{array}{rcl} &\{a_v~, ?p = attr_1 \land ?q = attr_2\}/ \\ &(\{a_e, \top\} / \{a_v~, ?p - attr_1 + ?q - attr_2 \le c 
                                                                                                                        \\ & \land attr_1 - ?q + ?p - attr_2 \le c 
                                                                                                                        \\ & \land attr_1 - ?q + attr_2 - ?p \le c  
                                                                                                                        \\ & \land ?q - attr_1 + attr_2 - ?p \le c \})^* \end{array} \)} & Complex\\
    \bottomrule
  \end{tabular}
\end{table*}

\paragraph{How we ran the experiments?}
The evaluation is conducted on a Ubuntu \texttt{22.04} LTS subsystem on a Intel 13700H laptop assigned with \qty{16}{GB} RAM. The experiments set a time-out threshold of \qty{10}{s} for each query for all data sets except the Telecom dataset, and the time-out threshold for the Telecom dataset is set to \qty{30}{s} due to its larger size and density.

\paragraph{Baseline} As a baseline we use the default implementation of RPQs in \textsc{MillenniumDB}. This means that for each query we generate, we have a baseline RPQ which simply removes all the data constraints. This will allow us to measure the effect of adding data constraints to RPQs.

\paragraph{Evaluation Metrics}
We evaluate the performance of the algorithms based on three metrics: (1) running time of each query;  
(2) general memory consumption: each run was executed against a dedicated, freshly started instance of~\textsc{MillenniumDB}, and the memory consumption was measured externally 
by monitoring the peak Resident Set Size (RSS) of \textsc{MillenniumDB} process during query execution; (3) oracle query memory consumption: we export the memory consumption 
of \textsc{Z3} during a query evaluation by its \texttt{C} API; and (4) oracle query count: the number of invocations of 
the SMT solver recorded in the logs of \textsc{MillenniumDB}, which we treat as queries to an external oracle, because it is a new feature introduced to graph path queries.

\subsection{Scalability Evaluation}
This subsection studies the \emph{feasibility} of the optimized algorithm (Algorithm~\ref{alg:query_algorithm}). For this, we test how the optimized algorithm performs across different queries, graph sizes and graph density; i.e. we answer research questions \textbf{RQ1}, \textbf{RQ2} and \textbf{RQ3}. We start by measuring the time performance of our solution over different queries and different graph sizes.

\paragraph{Time Performance}
We study the \emph{running-time distribution} of the optimized algorithm. Although our macro-state algorithm increases significantly rather than normal regular path queries,
the optimized algorithm still performs well across different datasets and queries.

Figure~\ref{fig:time-data} presents the running time distribution across datasets and data constraints, 
and Figure~\ref{fig:time-reg} presents the running time distribution across datasets and regular templates. We observe that the optimized algorithm maintains a low median 
running time across datasets, and distributions over sparse graphs are more concentrated than those over dense graphs, while dense graphs, although running time distributions are wider, 
most queries can still be evaluated within \qty{1}{s}. Such results indicate that \emph{the optimized algorithm scales well with graph density and size}, and provide a positive answer to 
\textbf{RQ1} from the perspective of time performance.

The performance variation with different queries is also studied and we obtain positive results.
According to the results in Figure~\ref{fig:time-reg}, we observe that the optimized algorithm performs well over frequently-used templates with most queries can be finished within 
\qty{1}{s} \emph{no matter whether the data constraint is simple or complex}, which cover nearly 90\% of real-world queries. According to the results in Figure~\ref{fig:time-data}, 
we observe that complex data constraints have better time performance than simple data constraints in terms of median running time and the concentrated distribution, while the simple 
data constraints have a longer tail in the distribution but most queries can still be evaluated within \qty{1}{s}. These results indicate that \emph{the optimized algorithm scales well with 
different queries}, and provide a positive answer to \textbf{RQ2} in terms of time performance.

\paragraph{Baseline comparison} While we showed that our optimized algorithm scales well over real-world graph, we still need to measure the 
impact of adding complex data comparisons to RPQs. For this, in Figure~\ref{fig:rpq-vs-prpq}, we measure the time of running the regular pattern 
of each of tested queries and compare with the time of running parametric regular path queries.
As can be observed, the optimized algorithm takes much more time compared to usual regular path queries by 10 times more average running time and wider distribution, 
as normal regular path queries can be evaluated within \qty{15}{ms} in most cases in Figure~\ref{fig:rpq-vs-prpq}. The reason for this is that the optimized 
algorithm does much more work since the data constraint portion of the query must be satisfied which requires an extensive amount of queries to the oracle 
during the evaluation. Next we measure the effect of oracle queries.


\begin{figure}[htbp]
  \centering
  \includegraphics[width=.3\textwidth]{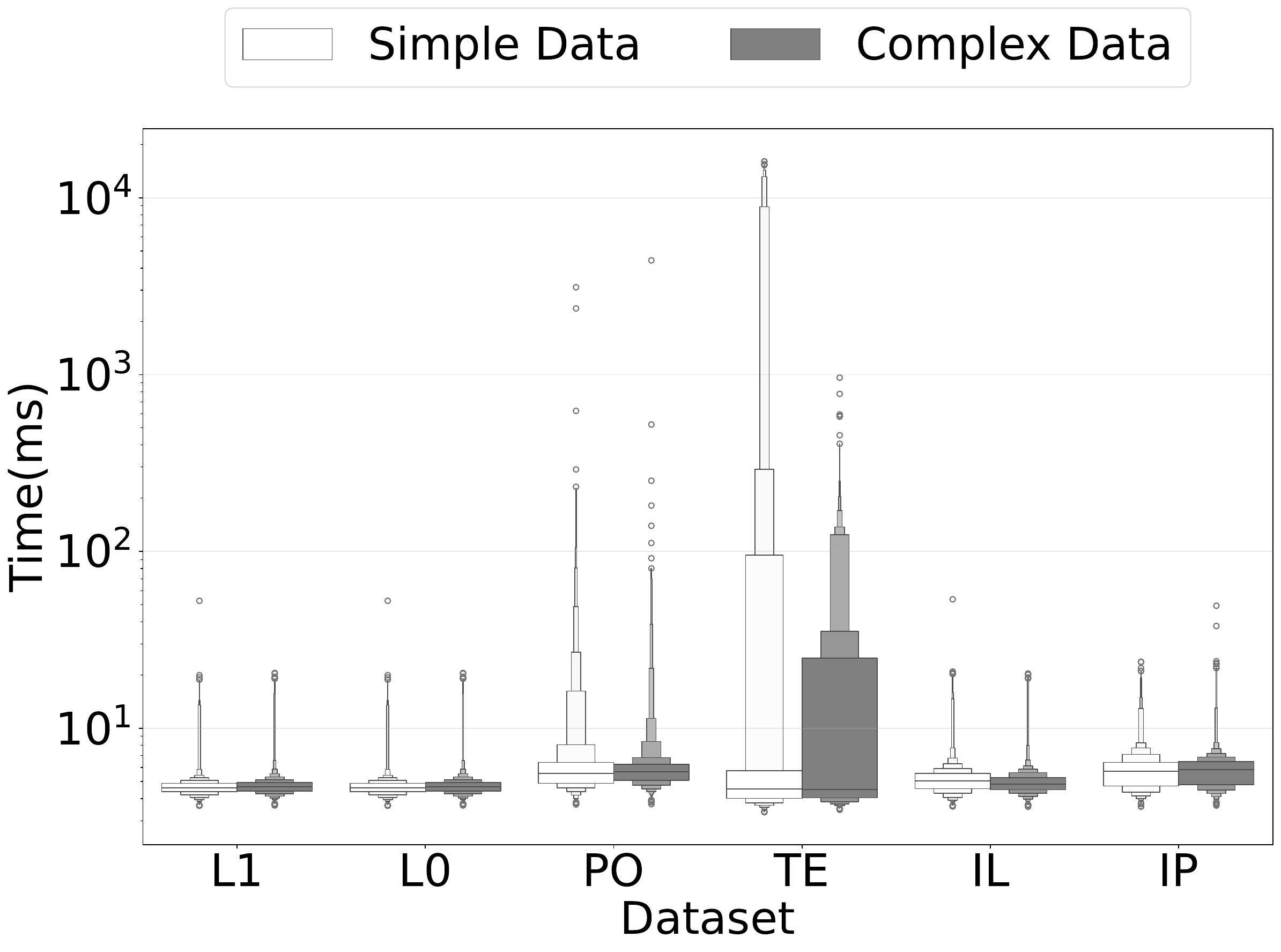}
  \caption{Running time distributions for queries across different datasets and data constraints in Table~\ref{tab:data_constraint}}
  \label{fig:time-data}
\end{figure}

\begin{figure}[htbp]
  \centering
  \includegraphics[width=.3\textwidth]{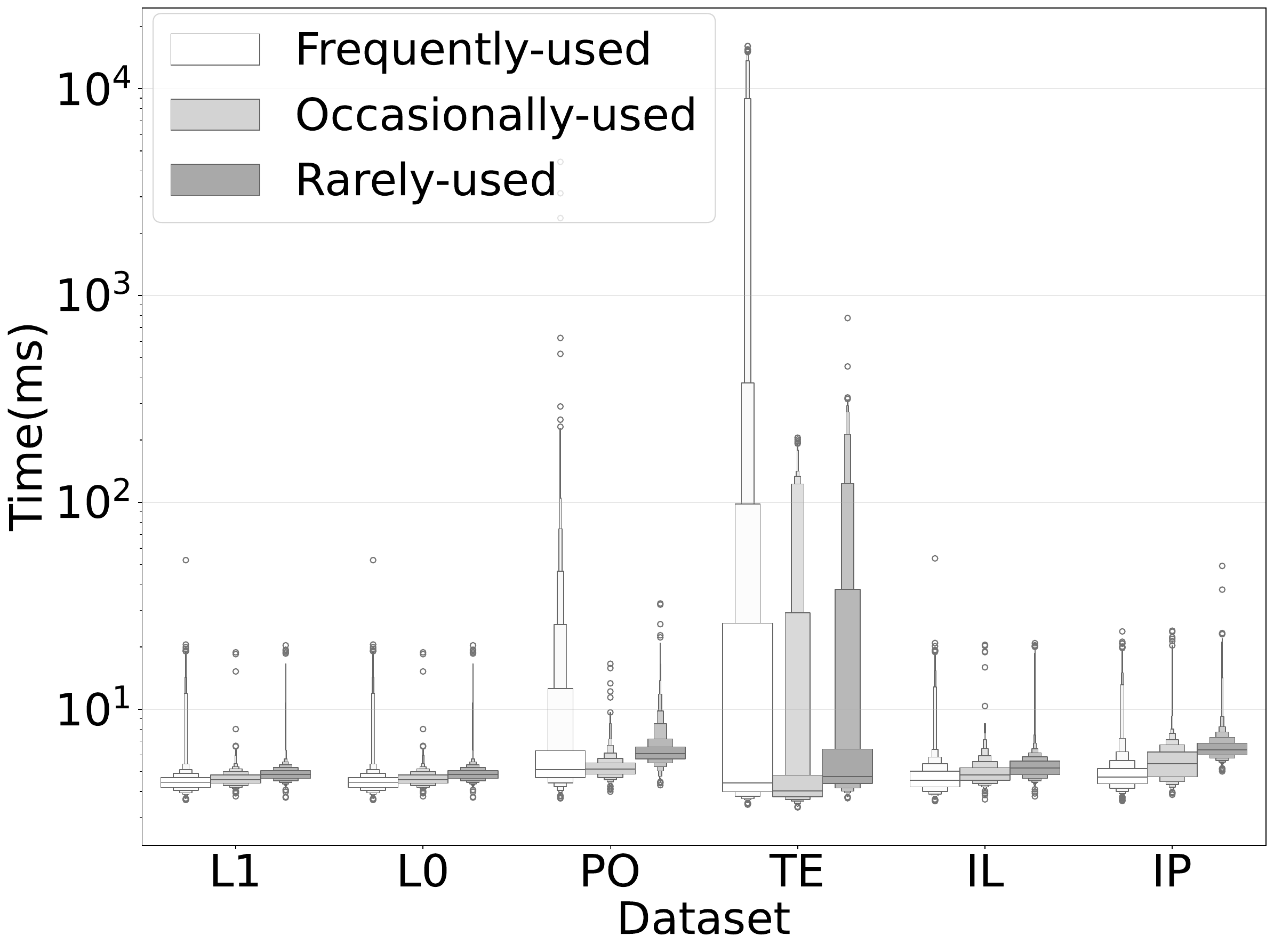}
  \caption{Running time distributions for queries across different datasets and regular templates in Table~\ref{tab:query_categories}}
  \label{fig:time-reg}
\end{figure}

\begin{figure}[htbp]
  \centering
  \includegraphics[width=.3\textwidth]{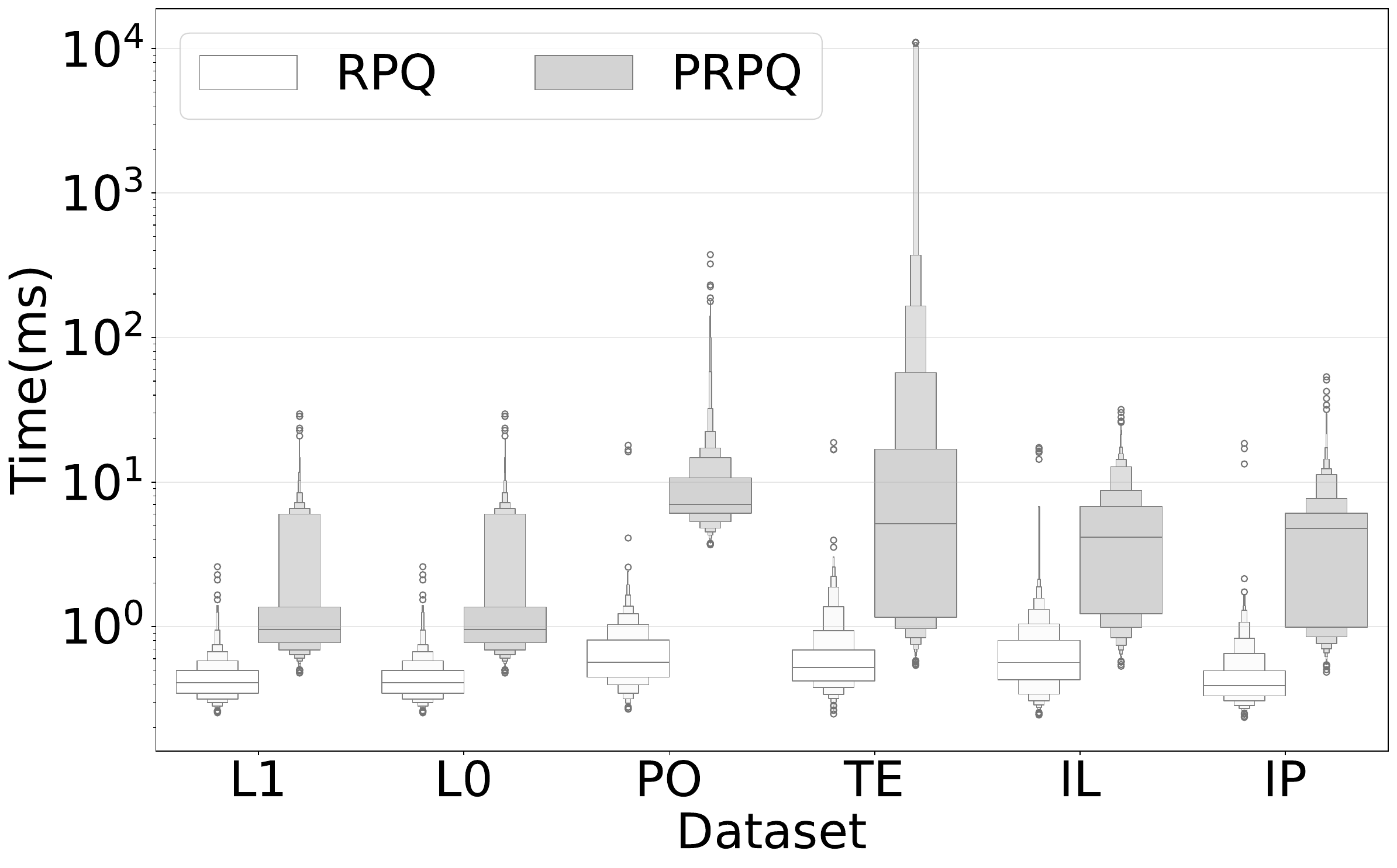}
  \caption{Running time of regular path queries and parametric regular path queries across datasets.}
  \label{fig:rpq-vs-prpq}
\end{figure}

\paragraph{Oracle Query Count}
From the view of complexity analysis, the optimized algorithm invokes the oracle with a length $c$ which is much smaller than the size of the graph $|V| + |E|$. As 
a result, the total count of oracle queries dominates the performance of a parametric regular path query. Statistical evidence in Figure~\ref{fig:oracle-count} shows that the 
running time of an individual query is strongly correlated to the oracle query count during the evaluation by a correlation coefficient 0.964 with $p < 0.001$.

\begin{figure}[htbp]
  \centering
  \includegraphics[width=.3\textwidth]{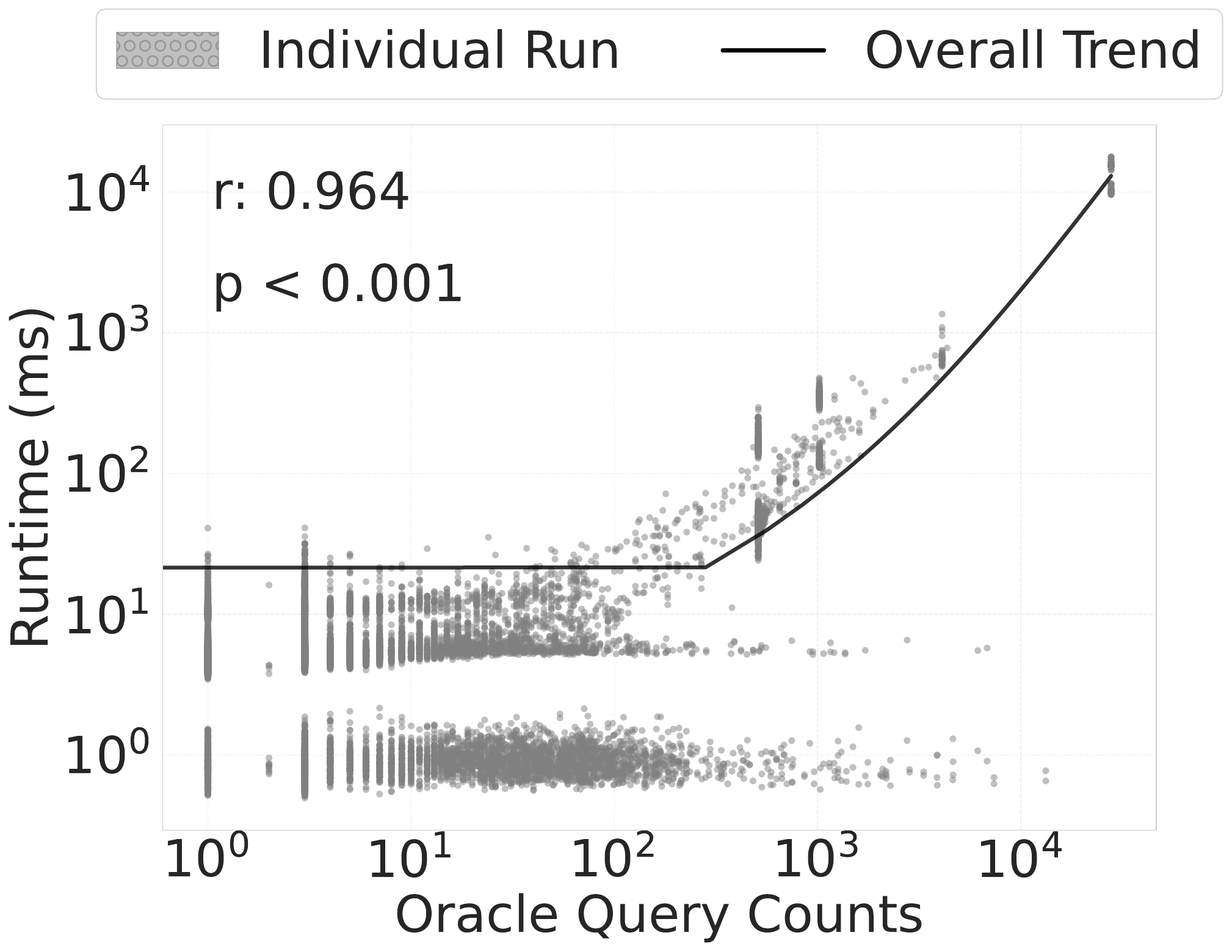}
  \caption{The distribution of each individual parametric regular query running time and its oracle query count}
  \label{fig:oracle-count}
\end{figure}

Figure~\ref{fig:oracle-count-data} presents the oracle query count distributions across data constraints and datasets. We observe that the parametric regular queries in dense graphs 
produce much more oracle queries compared to sparse graphs, while most queries among the dense datasets keep oracle query count below $10^3$, except the simple-data-constraint queries 
in TE. We also observe that the optimized algorithm queries the oracle more times for \emph{simple} data constraints, especially for large and dense graphs L1, TE and PO, which indicates 
that complex data constraints can terminate earlier, due to reaching a contradiction more easily.

According to Figure~\ref{fig:oracle-count-reg} and Figure~\ref{fig:time-reg}, we observe that the queries with the total oracle 
query count below $10^3$ can be evaluated within \qty{100}{ms} in most cases, and for the extreme cases with oracle query count above $10^4$, the running time can also be within 
\qty{10}{s}. With simple estimations, the average time for each oracle query is about the level of \qty{1}{ms}, which is acceptable for most applications. These results indicate that the 
optimized algorithm scales well with \emph{high oracle query counts}, and provide strong support for our complexity analysis and a positive answer to \textbf{RQ3}.

\begin{figure}[htbp]
  \centering
  \includegraphics[width=.3\textwidth]{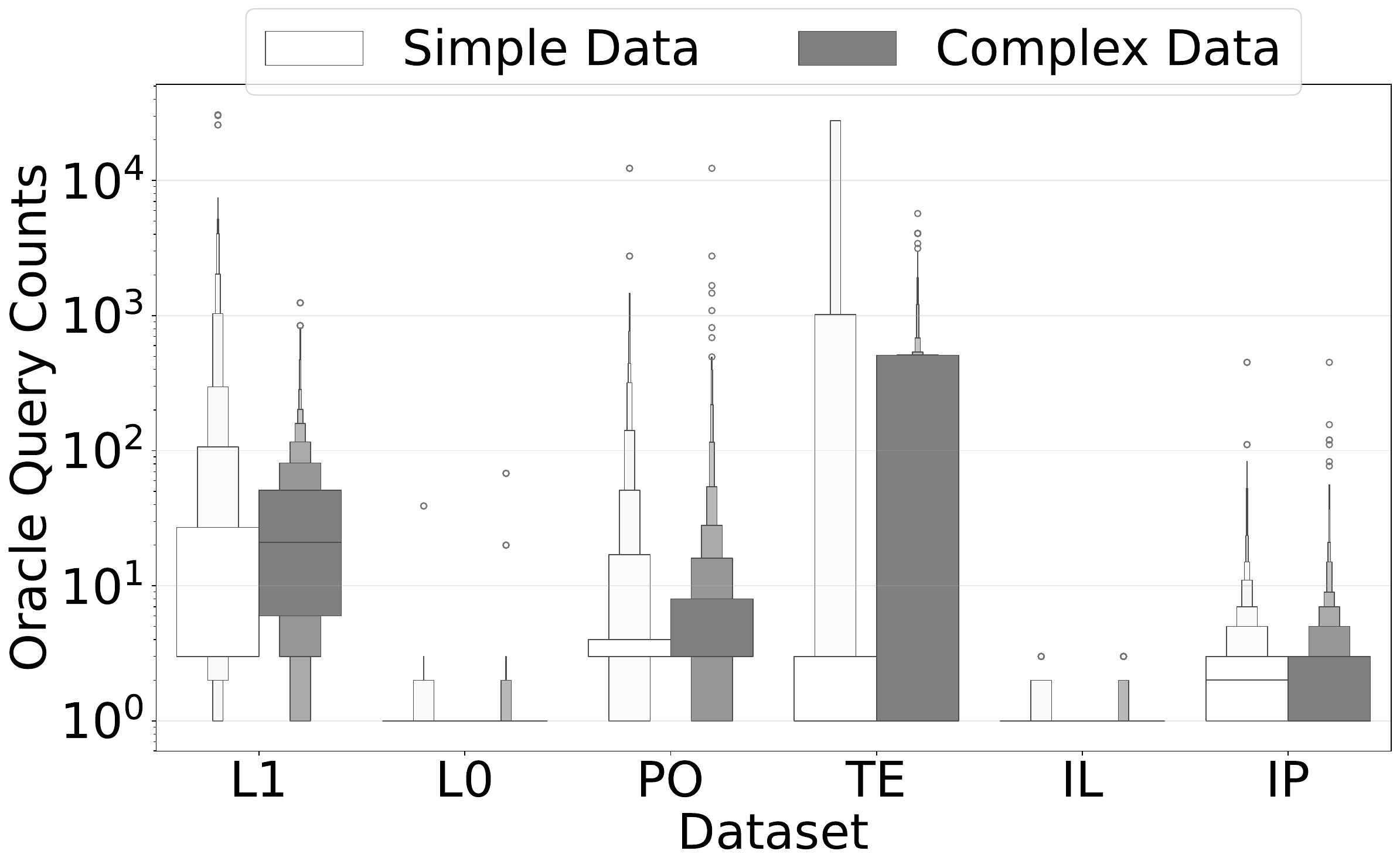}
  \caption{The distribution of oracle query counts across different data constraints in Table~\ref{tab:data_constraint}}
  \label{fig:oracle-count-data}
\end{figure}

Figure~\ref{fig:oracle-count-reg} presents the oracle query count distributions across regular templates and datasets. We observe that queries in frequently-used templates
produce the most oracle queries among all templates, but most queries only produce oracle queries below $10^3$, which indicates that 87.58\% of real-world queries can be still evaluated
efficiently. Both queries in occasionally-used and rarely-used templates produce less oracle queries, and the counts are moderate by most below $10^2$.

\begin{figure}[htbp]
  \includegraphics[width=.3\textwidth]{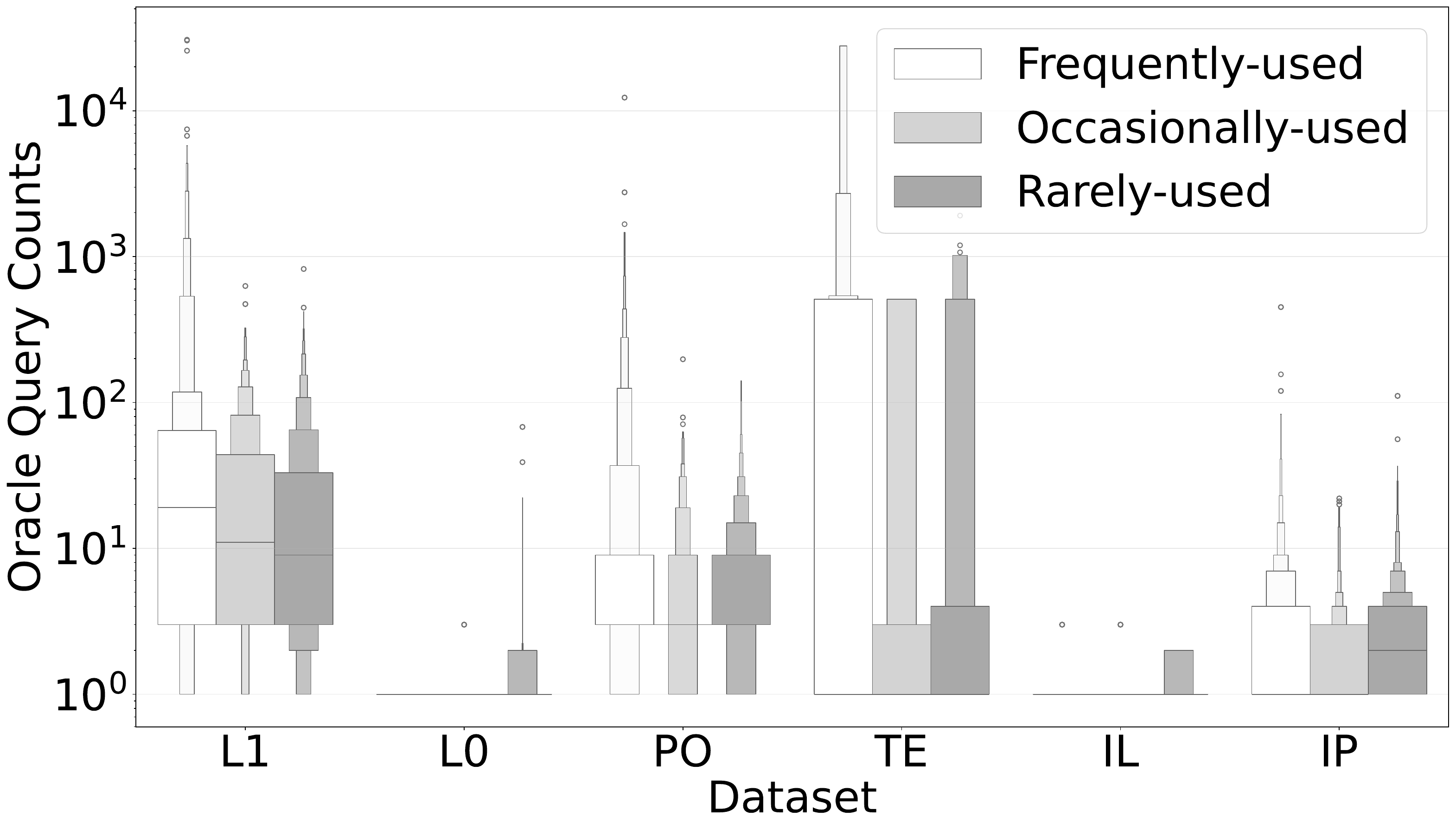}
  \caption{The distribution of oracle query counts across different regular templates in Table~\ref{tab:query_categories}. }
  \label{fig:oracle-count-reg}
\end{figure}

\paragraph{Memory Performance}
The memory consumption of the optimized algorithm is moderate but increases markedly during regular path queries (Figure~\ref{fig:memory-comparison}), which 
shows that the our macro-state approach is feasible in terms of memory consumption. 

We attribute this increase to the cost of querying the oracle (the SMT solver), i.e. $M_{oracle} = M_{prpq} - M_{rpq}$, where $M_{prpq}$ is the memory usage of 
a parametric regular path query, and $M_{rpq}$ is the memory usage of a regular path query with the same start point and regular template as $M_{prpq}$ respectively.
We measure the memory consumption of the SMT solver $M_{\textsc{Z3}}$, and calculate the increasing memory $\Delta M = M_{prpq} - M_{rpq}$. 

Table~\ref{tab:correlations_all} shows a strong linear correlation between the sample-level differences of $\Delta M$ and \textsc{Z3}'s memory usage for datasets 
L1, L0, PO, IL, and IP ($r = 0.963$--$1.000$, $p < 0.001$). The TE dataset was an exception, showing moderate correlations for D1 and D2 ($r = 0.454$--$0.470$, $p < 0.001$) 
but strong correlations for D3-D5 ($r = 0.869$--$0.977$, $p < 0.001$), indicating variability in early measures. The statistical analysis confirms that oracle queries are the principal
cause of increased memory usage.

However, we have not observed a strong correlation between the memory consumption of oracle queries and graph size or density, or variation of queries and data constraints, 
which indicates that the memory consumption of oracle queries is mainly determined by the SMT solver itself rather than the graph or queries, which provides an answer to 
\textbf{RQ1}, \textbf{RQ2} and \textbf{RQ3} from the perspective of memory consumption.

\begin{figure}[htbp]
  \centering
  \includegraphics[width=.25\textwidth]{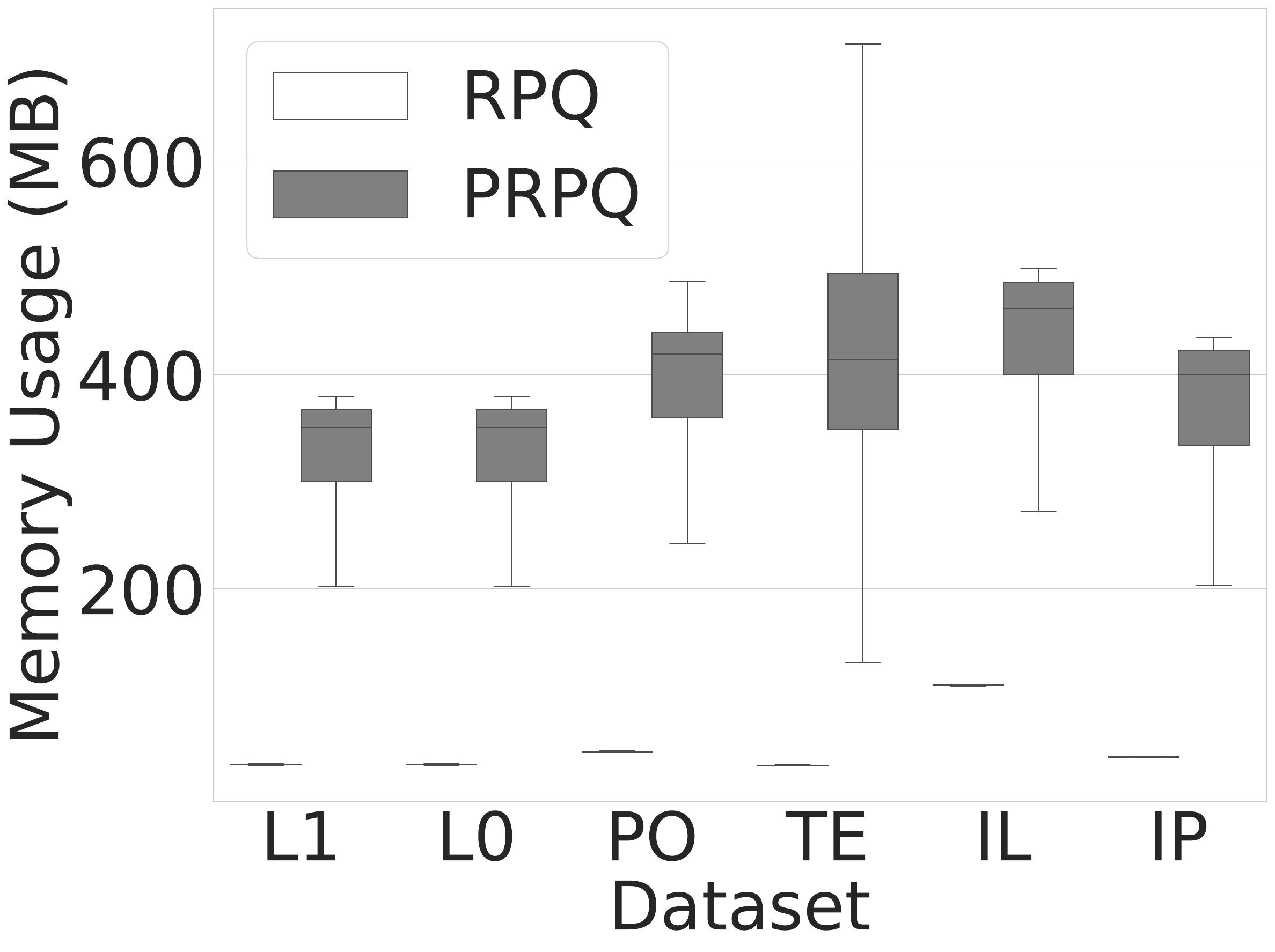}
  \caption{The distribution of parametric regular path queries and normal regular path queries memory consumption.}
  \label{fig:memory-comparison}
\end{figure}

\begin{table}[htbp]
\centering
\caption{Sample-level correlation between the memory usage differences $\Delta M$ and the memory consumption of \textsc{Z3} with $p < 0.001$}

\begin{tabular}{l c c c c c}
  \hline
\textbf{Dataset} & \textbf{D1} & \textbf{D2} & \textbf{D3} & \textbf{D4} & \textbf{D5} \\
\hline
L1 & 0.970 & 0.969 & 0.971 & 0.963 & 0.967 \\ 
L0 & 0.999 & 0.999 & 0.999 & 0.999 & 0.999 \\
PO & 0.995 & 0.995 & 1.000 & 1.000 & 1.000 \\
TE & 0.454 & 0.470 &  0.869 & 0.977 & 0.946 \\ 
IL & 0.998 & 0.998 & 0.998 & 0.998 & 0.996 \\ 
IP & 0.998 & 0.998  &0.998 & 0.998 & 0.995 \\
\hline
\end{tabular}
\label{tab:correlations_all}
\end{table}

\paragraph{Conclusions} Overall, even if the parametric regular path queries evaluation is NP-hard according to Theorem~\ref{theorem:np_hardness}, the experimental results demonstrate the \emph{feasibility} of 
the optimized algorithm (Algorithm~\ref{alg:query_algorithm}) which scales well with large and dense datasets and complex data constraints.





\subsection{Performance Comparison}
This subsection foucuses on \textbf{RQ4} and compares the performance of the naive algorithm (Algorithm~\ref{alg:naive_algorithm}) and the optimized algorithm (Algorithm~\ref{alg:query_algorithm}), 
showing a significant gain in deploying the optimizations described in Section~\ref{sec:algos}.

\paragraph{Running Time Analysis}
The optimized algorithm significantly improves the running time for queries that both algorithms can complete within the time limit. Figure~\ref{fig:time-compare} shows the running time 
distributions of both algorithms across all datasets and data constraints (for queries that do not time out). We observe that the optimized algorithm
has lower median running time and more concentrated distribution compared to the naive algorithm for complex data constraints across the dense and large-size graphs. Although
the performance of the naive algorithm is slightly better to that of the optimized algorithm on sparse graphs such as IL, the difference is marginal, and the general 
running time of optimized algorithm still remains fast.
\begin{figure}[htbp]
  \centering
  \caption{Running time across datasets and data constraints in Table~\ref{tab:data_constraint}}
  \label{fig:time-compare}
 \begin{subfigure}{.15\textwidth}
  \caption{L1}
  \includegraphics[width=\textwidth]{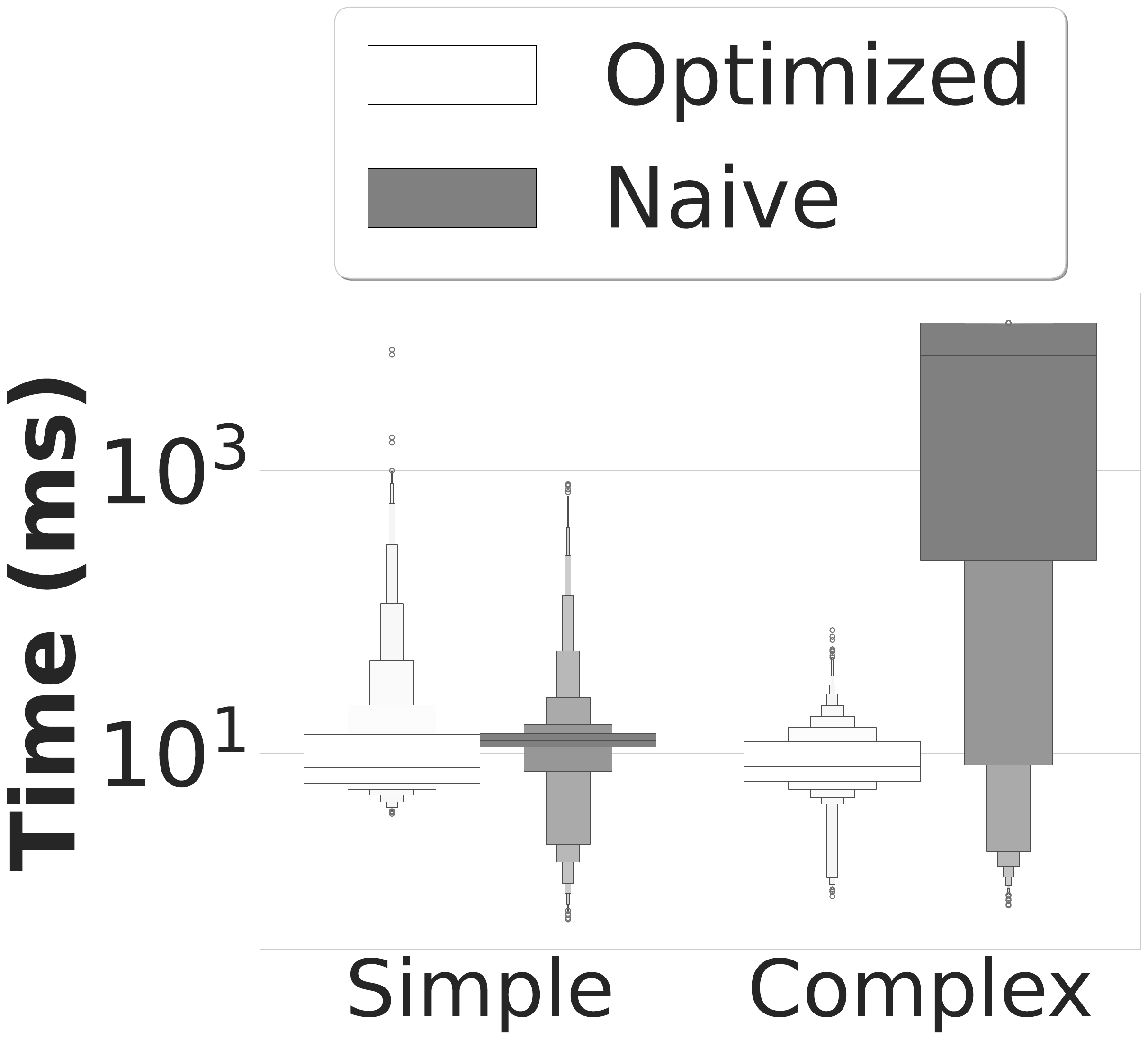}
  \label{fig:ldbc10-timeout}
 \end{subfigure}
  \begin{subfigure}{.15\textwidth}
  \caption{PO}
  \includegraphics[width=\textwidth]{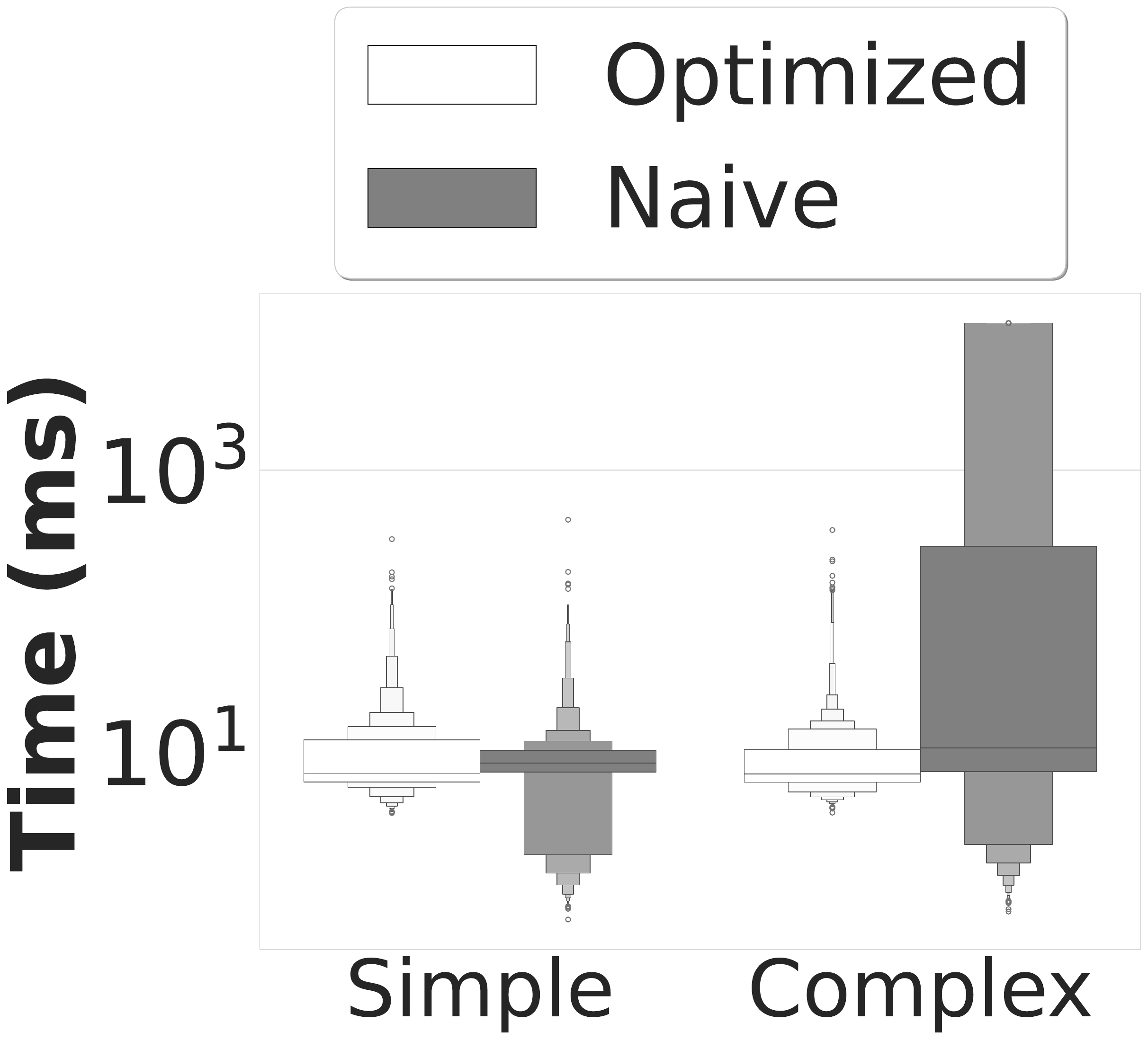}
  \label{fig:pokec-timeout}
 \end{subfigure}
  \begin{subfigure}{.15\textwidth}
  \caption{TE}
  \includegraphics[width=\textwidth]{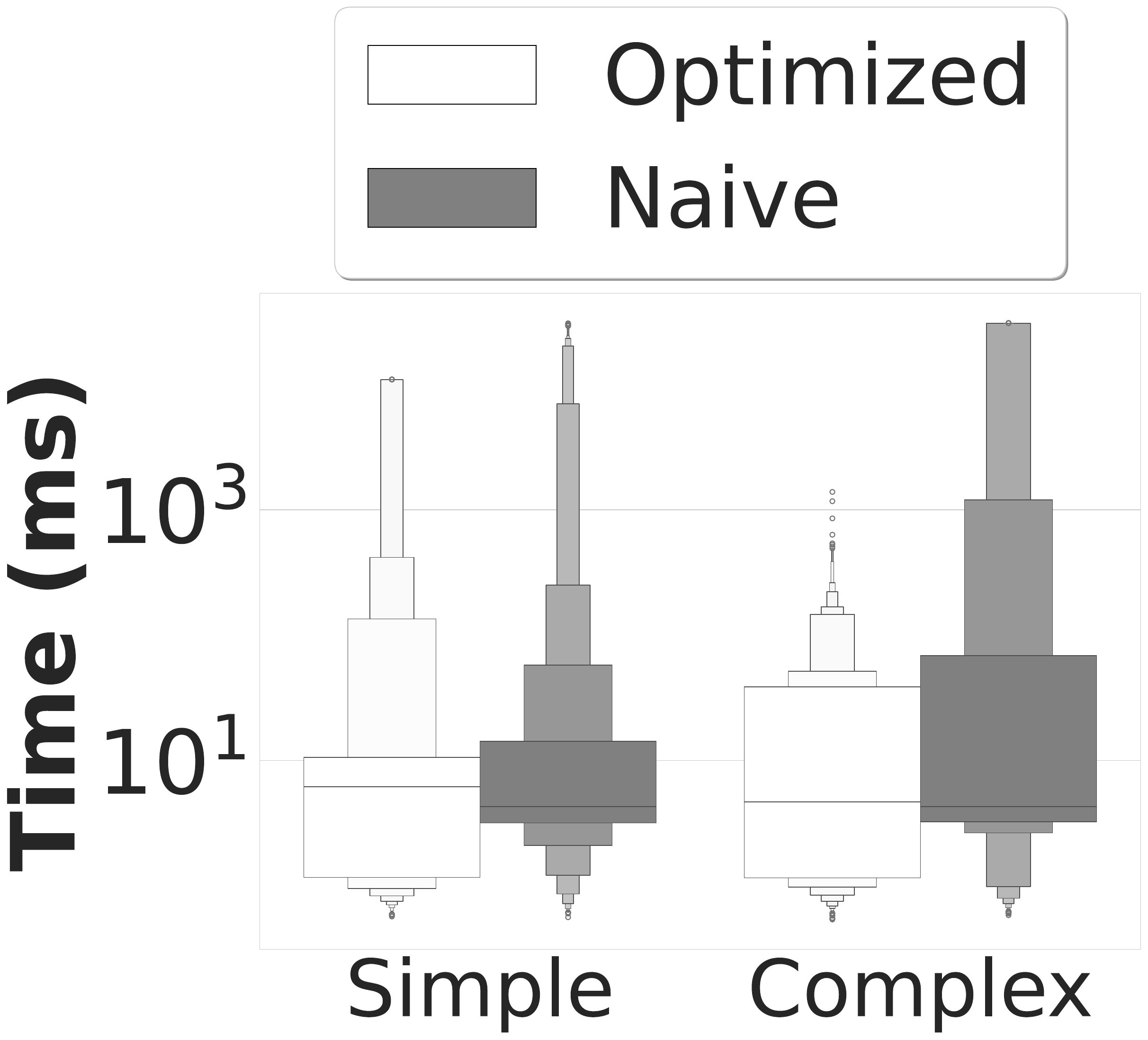}
  \label{fig:telecom-timeout}
 \end{subfigure}
  \begin{subfigure}{.15\textwidth}
  \caption{L0}
  \includegraphics[width=\textwidth]{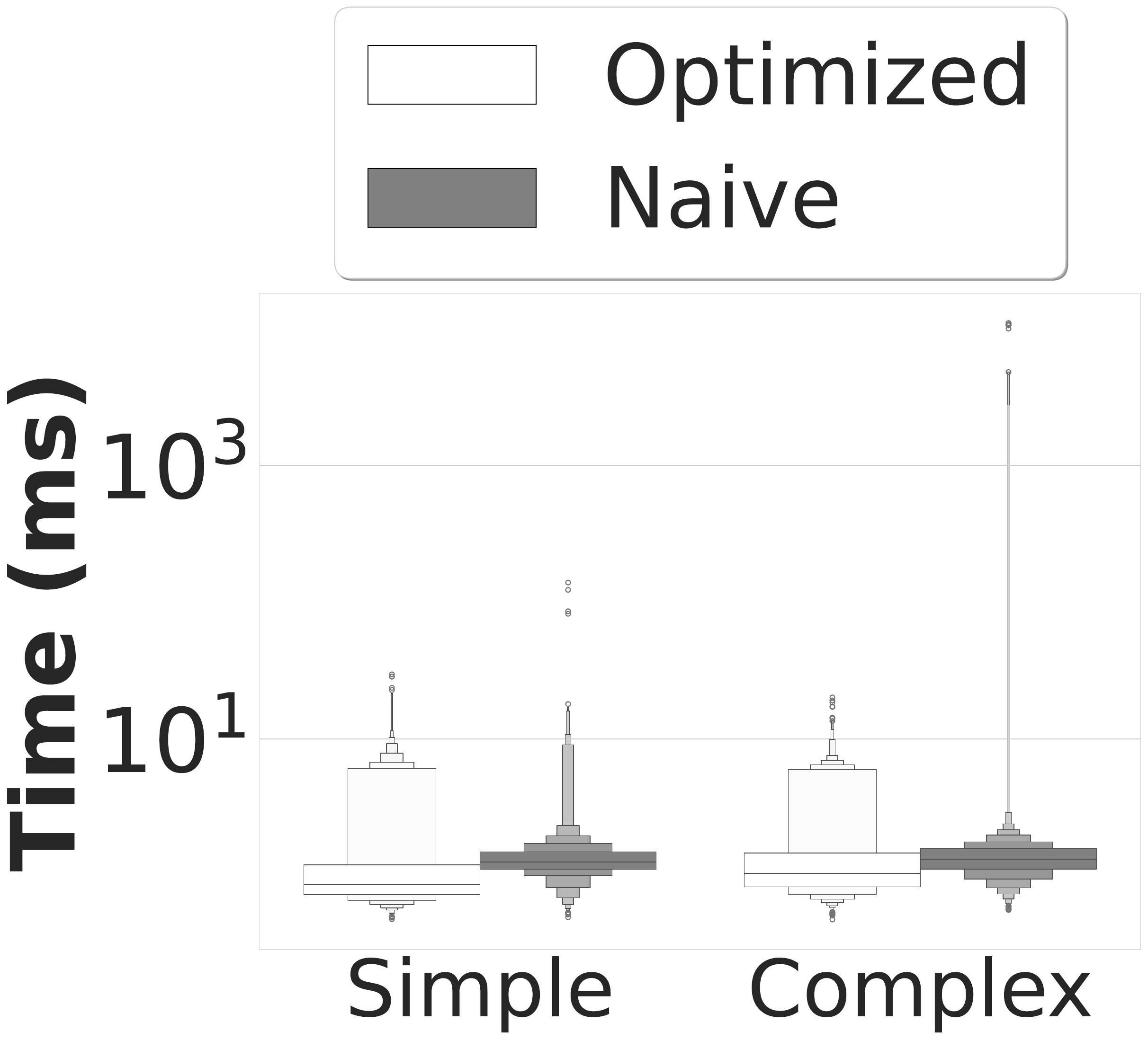}
  \label{fig:ldbc01-timeout}
 \end{subfigure}
  \begin{subfigure}{.15\textwidth}
  \caption{IL}
  \includegraphics[width=\textwidth]{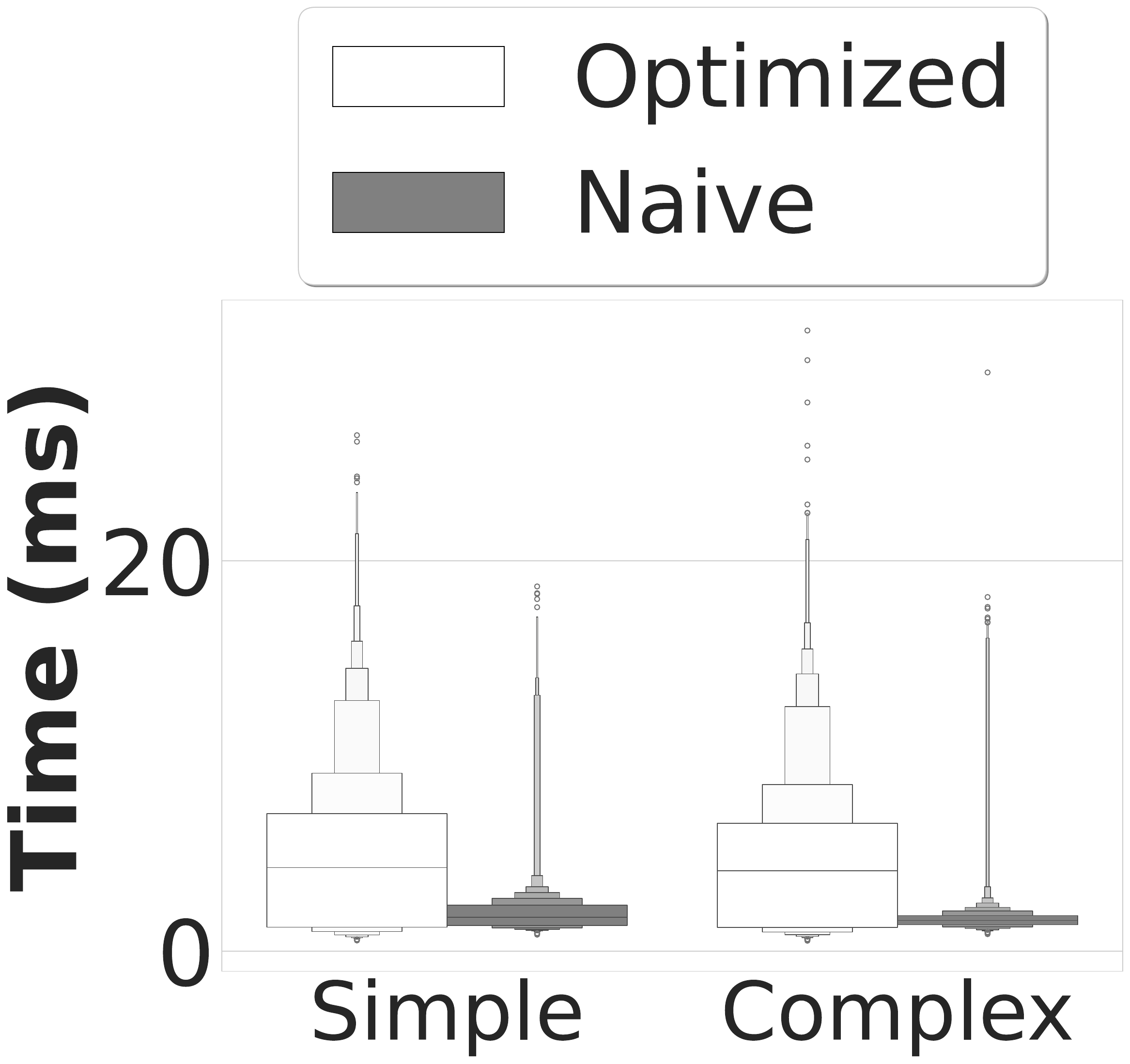}
  \label{fig:icij-Leaks-timeout}
 \end{subfigure}
  \begin{subfigure}{.15\textwidth}
  \caption{IP}
  \includegraphics[width=\textwidth]{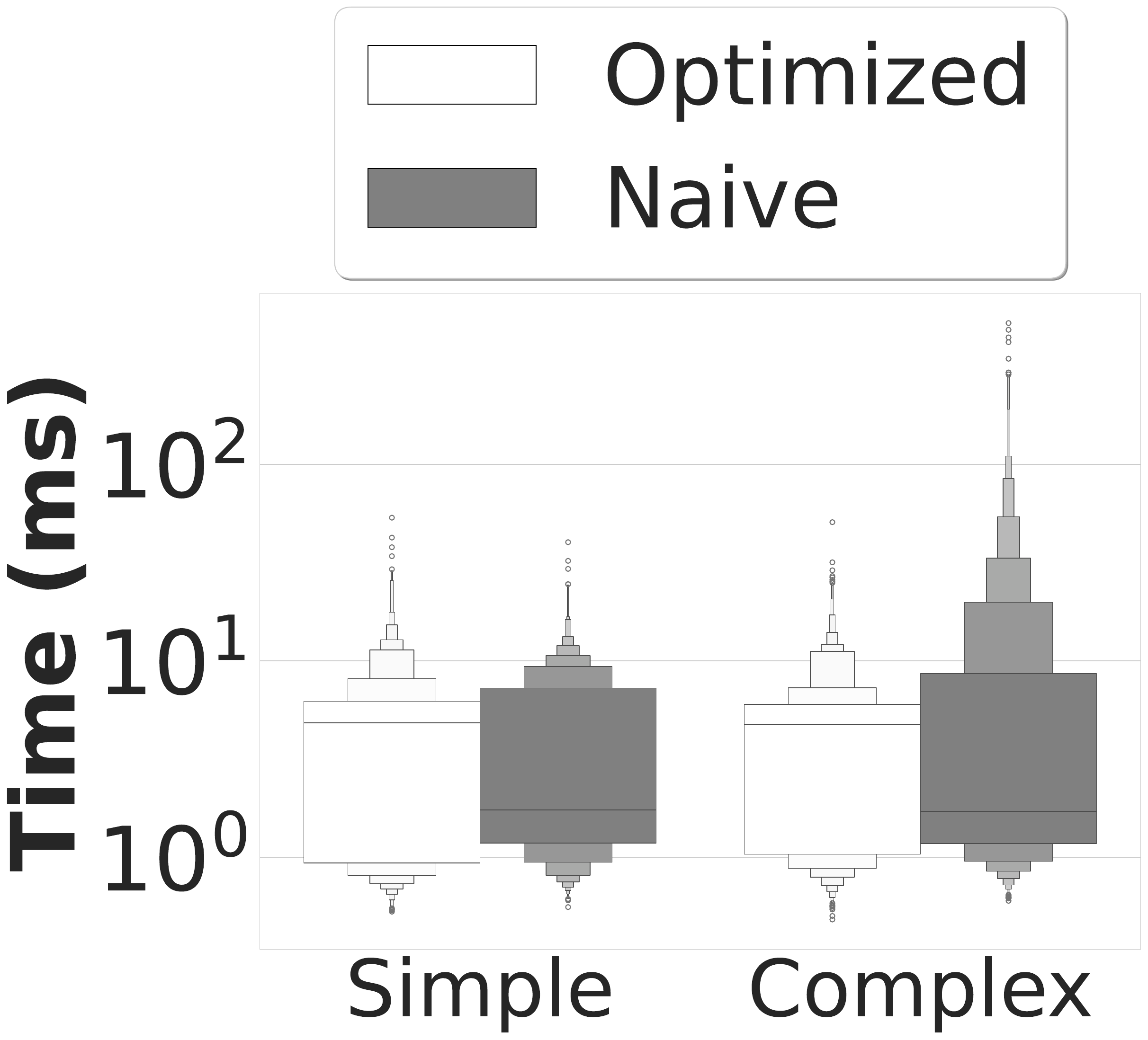}
  \label{fig:Paradises-timeout}
 \end{subfigure}
\end{figure}

\paragraph{Time-out Analysis} A significant gain of the optimized algorithm can be observed with respect to the number of queries that time out. That is, the optimized algorithm 
reduces the time-out rate considerably compared to the naive algorithm. Figure~\ref{fig:time-reduction} presents the time-out rates of naive and optimized algorithm across data constraints and datasets where timeouts occur. We observe that \emph{optimized algorithm 
improves the time-out markedly for complex data constraints}, where the naive algorithm fails to complete a markedly larger portion of queries within the time limit, while \emph{most} 
queries are successfully evaluated by the optimized algorithm within the prescribed timeout threshold for each dataset.

\begin{figure}[htbp]
  \centering
  \caption{Time-out rate across datasets and data constraints in Table~\ref{tab:data_constraint}}
  \label{fig:time-reduction}
 \begin{subfigure}{.15\textwidth}
  \caption{L1}
  \includegraphics[width=\textwidth]{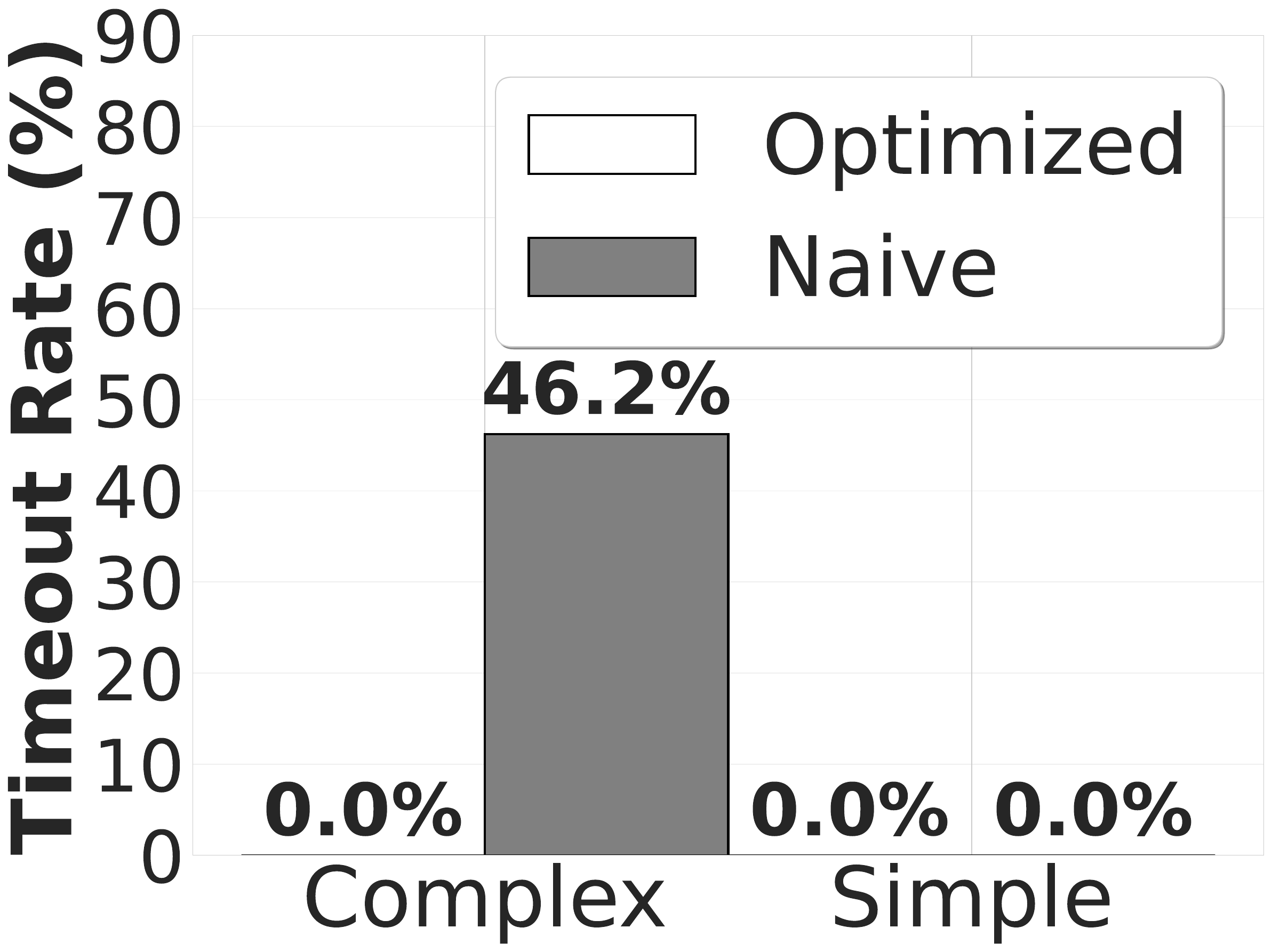}
  \label{fig:ldbc10-timeout}
 \end{subfigure}
  \begin{subfigure}{.15\textwidth}
  \caption{PO}
  \includegraphics[width=\textwidth]{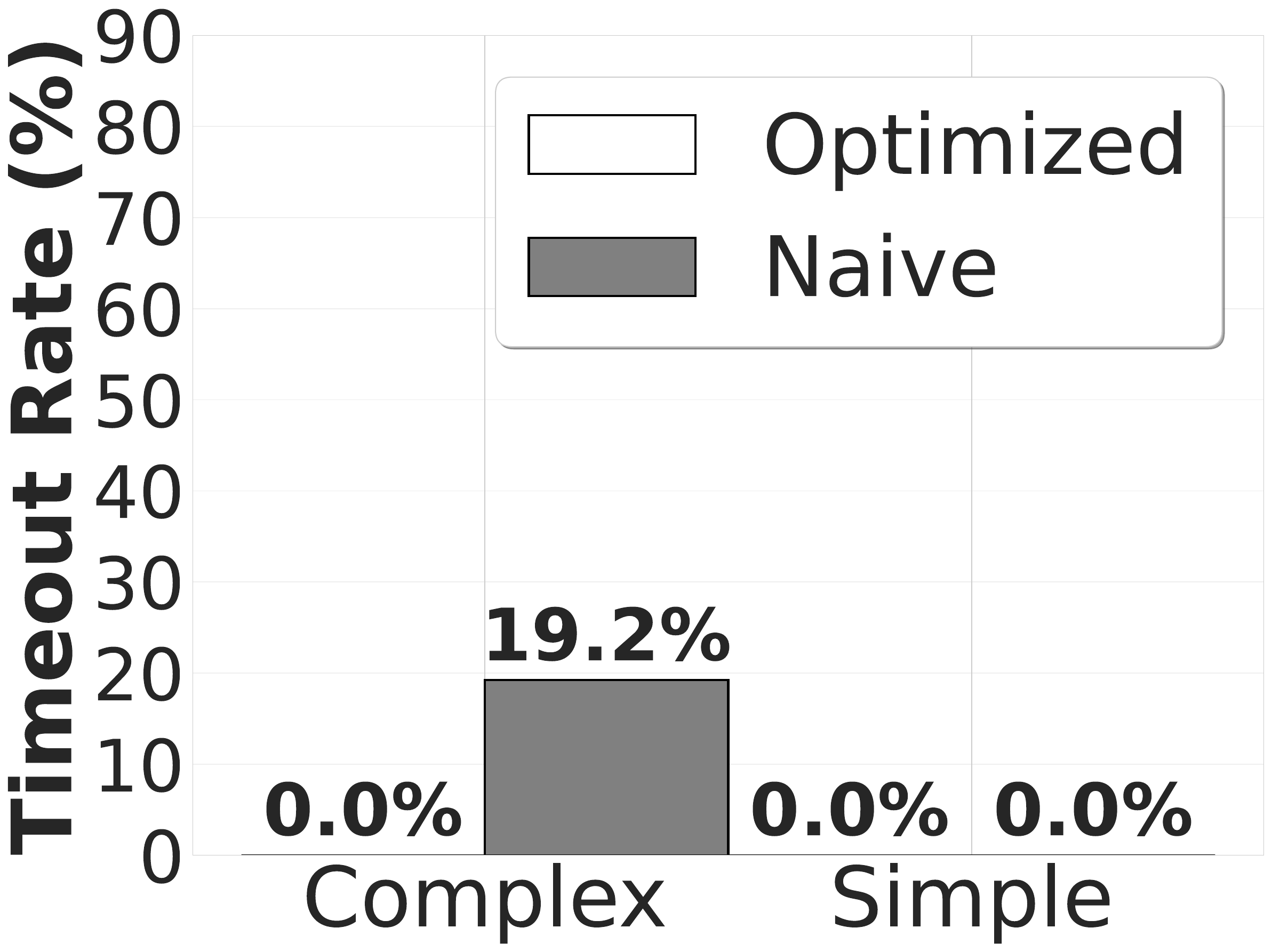}
  \label{fig:pokec-timeout-1}
 \end{subfigure}
  \begin{subfigure}{.15\textwidth}
  \caption{TE}
  \includegraphics[width=\textwidth]{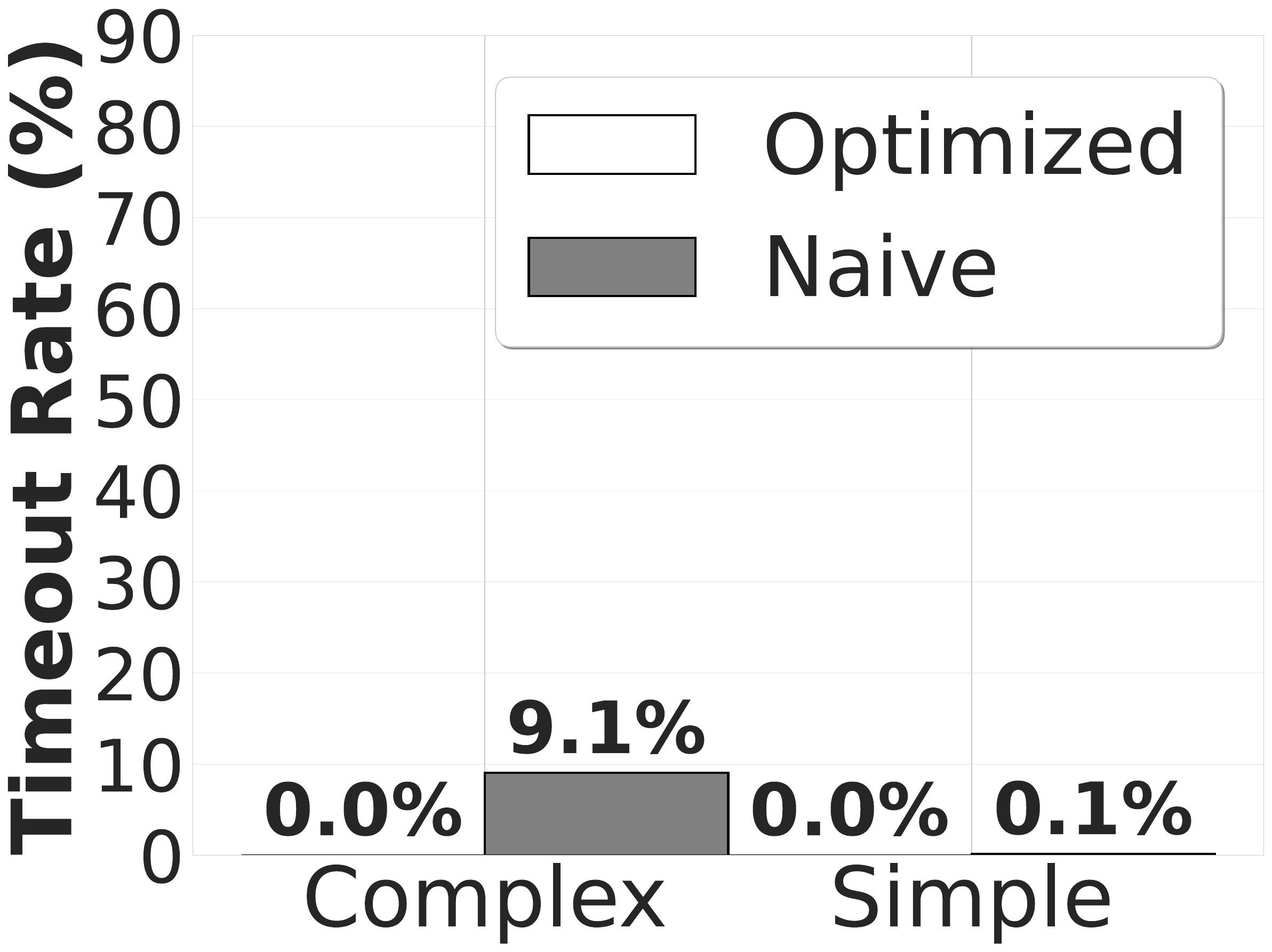}
  \label{fig:telecom-timeout}
 \end{subfigure}
\end{figure}
                 

Overall, these results provides positive evidence to \textbf{RQ4}, demonstrating that \emph{the optimized algorithm significantly 
reduces the time-out rate and improves the running time, especially for complex data constraints}.


\section{Conclusions}
\label{sec:concl}
In this paper we extend navigational query languages for graph databases with
the constraints for reasoning on how data values change along the explored 
paths. To this end, we introduced parametric regular path queries and developed
efficient algorithms for evaluation of these queries. 
We implemented our solution on top of MillenniumDB --- an open-source
industry-strength graph database system that fully supports regular path queries
--- showing its feasibility in real-world scenarios. Furthermore, our 
implementation of the algorithm constitutes the first scalable constraint 
database system, which can handle databases with up to tens of millions of 
edges. There are many avenues for future work, in particular supporting more
intricate constraints. Although our algorithm can be extended to handle
linear integer arithmetic, it does not immediately extend to 
\emph{nonlinear} real arithmetic. In particular, \emph{can we achieve this
without 
the heavy
machinery of constraint database theory?}
We leave this as an open problem.

\clearpage
\bibliographystyle{ACM-Reference-Format}
\bibliography{sample}
\clearpage
\end{document}